\documentclass{article}

\usepackage{graphicx,epic,eepic,epsfig,amsmath,amsfonts,latexsym,amssymb,color,enumerate } 

\usepackage{mathtools}
\usepackage{fullpage}

\newtheorem{theorem}{Theorem}
\newtheorem{lemma}[theorem]{Lemma}

\newtheorem{definition}[theorem]{Definition}
\newtheorem{proposition}[theorem]{Proposition}
\newtheorem{corollary}[theorem]{Corollary}

\def\squareforqed{\hbox{\rlap{$\sqcap$}$\sqcup$}}
\def\qed{\ifmmode\squareforqed\else{\unskip\nobreak\hfil
\penalty50\hskip1em\null\nobreak\hfil\squareforqed
\parfillskip=0pt\finalhyphendemerits=0\endgraf}\fi}
\def\endenv{\ifmmode\;\else{\unskip\nobreak\hfil
\penalty50\hskip1em\null\nobreak\hfil\;
\parfillskip=0pt\finalhyphendemerits=0\endgraf}\fi}
\newenvironment{proof}[1][Proof]{\noindent\textbf{Proof.} }{\hfill\qed}

\usepackage{booktabs} 
\usepackage{array} 
\usepackage{caption} 

\DeclareMathOperator{\var}{Var}

\newcommand{\beq}{\begin{equation}}
\newcommand{\eeq}{\end{equation}}
\newcommand{\mc}{\mathcal}
\newcommand{\ox}{\otimes}

\newcommand*{\id}{\mathrm{id}}
\newcommand*{\tr}[1]{\mathrm{Tr}\left[#1\right]}

\newcommand*{\ket}[1]{| #1 \rangle}
\newcommand*{\bra}[1]{\langle #1 |}

\newcommand*{\proj}[1]{|#1\rangle\!\langle #1|}

\newcommand*{\mge}{\succcurlyeq}
\newcommand*{\mle}{\preccurlyeq}

\newcommand*{\exs}[2]{\mathbb{E}_{#1}\left[#2 \right]}

\newcommand*{\bid}{\mathbf{1}}

\newcommand*{\cA}{\mathcal{A}}

\newcommand*{\cE}{\mathcal{E}}

\newcommand*{\cG}{\mathcal{G}}
\newcommand*{\cH}{\mathcal{H}}

\newcommand*{\cN}{\mathcal{N}}
\newcommand*{\cM}{\mathcal{M}}

\newcommand*{\cO}{\mathcal{O}}
\newcommand*{\cP}{\mathcal{P}}

\newcommand*{\cS}{\mathcal{S}}
\newcommand*{\fS}{\mathfrak{S}}
\newcommand*{\cT}{\mathcal{T}}
\newcommand*{\cU}{\mathcal{U}}

\newcommand*{\cX}{\mathcal{X}}
\newcommand*{\cY}{\mathcal{Y}}

\newcommand{\supp}{{\operatorname{supp}}}
\usepackage{hyperref}

\usepackage{authblk}

\usepackage[maxnames = 99]{biblatex} 
\title{Exponents for classical-quantum channel simulation \\in purified distance}

\author[]{Aadil Oufkir}
\author[]{Yongsheng Yao}
\author[]{Mario Berta}

\affil[]{\small{Institute for Quantum Information,
  RWTH Aachen University,
  Aachen, Germany}}

\begin{document}

\maketitle

\begin{abstract}
    We determine the exact error and strong converse exponent for entanglement-assisted classical-quantum channel simulation in worst case input purified distance. The error exponent is expressed as a single-letter formula optimized over sandwiched R\'enyi divergences of order $\alpha \in [1, \infty)$, notably without the need for a critical rate\,---\,a sharp contrast to the error exponent for classical-quantum channel coding. The strong converse exponent is expressed as a single-letter formula optimized over sandwiched R\'enyi divergences of order $\alpha\in [\frac{1}{2},1]$. As in the classical work [Oufkir {\it et al.}, arXiv:2410.07051], we start with the goal of asymptotically expanding the meta-converse for channel simulation in the relevant regimes. However, to deal with non-commutativity issues arising from classical-quantum channels and entanglement-assistance, we critically use various properties of the quantum fidelity, additional auxiliary channel techniques, approximations via Chebyshev inequalities, and entropic continuity bounds.
\end{abstract}


\section{Introduction}

\subsection{Motivation}

Channel simulation is a fundamental information processing task. Dual to the noisy channel coding problem, which focuses on simulating noiseless channels using noisy channels, channel simulation aims to simulate noisy channels using noiseless channels. Bennett {\it et al.}~\cite{BSST2002entanglement} proved that under shared randomness assistance, the minimal classical communication cost rate to achieve asymptotically reliable classical channel simulation equals to the channel's capacity. This result is a milestone in classical information theory known as the reverse Shannon theorem. Later, Bennett {\it et al.}~\cite{BDHSW2014quantum} developed the quantum reverse Shannon theorem, which states that the minimal classical communication rate for asymptotically perfect quantum channel simulation under free entanglement assistance is given by the the entanglement-assisted classical capacity (see also the work of Berta {\it et al.}~\cite{BCR2011the} based on one-shot information theory).

In recent work, \cite{fang2019quantum,cao2024channel} showed that the minimal amount of noiseless classical communication needed to achieve one-shot non-signaling assisted classical channel simulation and one-shot non-signaling assisted quantum channel simulation is characterized by the channel's partially smoothed max-information. Furthermore, the minimal classical communication cost rate for asymptotically perfect non-signaling assisted channel simulation can be derived from these one-shot results. Although the first-order asymptotics for channel simulation under the assistance of various resources have been well understood, much less is known about the large-deviation type exponential behaviors for channel simulation. The work by  \cite{LiYao2021reliable} provides a significant step in this direction by deriving both  lower and upper  bounds  for the error exponent for entanglement-assisted quantum channel simulation. When the classical communication cost rate is below a certain critical rate, the lower and upper bounds coincide, thus determining the exact error exponent in the low-rate regime.

In the classical setting, the recent work by \cite{Oufkir2024Oct} establishes the error and strong converse exponents, under the total variation distance, and for all rates. In contrast, in the quantum setting, the purified distance has the advantage that it is phrased in terms of the quantum fidelity, which features Uhlmann's theorem \cite{uhlmann1976transition} that is an essential tool in quantum information theory. The purified distance in particular, also satisfies desirable properties such as the triangle inequality, the data processing inequality, and the achievability by quantum measurements~\cite{tomamichel2010duality, tomamichel2012framework, tomamichel2015quantum, wilde2013quantum}. For these reasons, we choose here to focus on exponents of channel simulation under the purified distance. Interestingly, the exponents we find for the purified distance, particularly the strong converse exponent, are quite different from those for the total variation distance in the classical setting proved by \cite{Oufkir2024Oct}. The techniques we employ are also notably different (and work in the classical-quantum setting as well).\\

\textbf{Note added.} Concurrently and independently, the exponents for classical channel simulation are derived by Li {\it et al.}~\cite{Li2024Oct} for general $\alpha$-R\'enyi based distances. This includes in particular the $\alpha=1/2$ case, which corresponds to log-fidelity and is closely related to the purified distance considered here.


\subsection{Overview of results}\label{sec:overview}

In this paper, we study the error exponent and the strong converse exponent for channel simulation under the purified distance, which represent the best rates of exponential convergence of the error of this task towards the perfect and worst-case performance, respectively. Specifically, we derive the exact error exponents for classical-quantum channel simulation under the assistance of non-signaling and entanglement resources and the exact error exponents for classical channel simulation under the assistance of non-signaling and shared randomness resources. All of them  are given by the following expression (see Proposition \ref{prop:EE-Ach-P-EA})
\begin{equation}\label{eq:res-EE}
 \frac{1}{2} \sup_{\alpha \geq 0} \alpha \big(r-\widetilde{I}_{1+\alpha}(W)   \big),
\end{equation}
where $W$ is a classical-quantum channel or a classical channel, $r$ is the classical communication cost rate and $\widetilde{I}_{\beta}(W)$ is the the channel's sandwiched R\'enyi mutual information of order $\beta$, i.e.,
\begin{equation}
    \widetilde{I}_{\beta}(W)\coloneqq\inf_{\sigma \,:\, \text{state}} \max_{x} \widetilde{D}_{\beta}(W_x\|\sigma)
\end{equation}
for all $\beta\in (0, \infty)$, and $\widetilde{D}_{\beta}$ is the sandwiched R\'enyi divergence of order $\beta$, defined as 
\begin{equation}
    \widetilde{D}_{\beta}(\rho \| \sigma)\coloneqq\frac{1}{\beta-1} \log \operatorname{Tr}\big({\sigma}^{\frac{1-\beta}{2\beta}} \rho {\sigma}^{\frac{1-\beta}{2\beta}}\big)^\beta \qquad \forall \beta \in (0,\infty). 
\end{equation}
Compared to the previous work~\cite{LiYao2021reliable}, our characterization of the error exponents does not exhibit a critical rate. Hence, we fill in the gap for the error exponent in the high-rate region for classical-quantum channel simulation. We conjecture that a critical rate is not needed for the error exponent in general quantum channel simulation and leave this as an open question for future investigation. 

In addition, we establish the exact strong converse exponents for shared randomness assisted (resp.~entanglement assisted) and non-signaling assisted classical (resp.~classical-quantum) channel simulation, which are characterized  by the channel's sandwiched R\'enyi mutual information as follows (see Corollaries \ref{cor:SCE-ach} and \ref{cor:SCE-ach-CQ})
\begin{equation}\label{eq:res-SCE}
\sup_{\frac{1}{2} \leq \alpha \leq 1} \frac{1-\alpha}{\alpha} \big(\widetilde{I}_{\alpha}(W)-r   \big).
\end{equation}
We refer to Table~\ref{tab:results} for a summary of our results. We derive our main results \eqref{eq:res-EE} and \eqref{eq:res-SCE} using the rounding technique recently developed in \cite{berta2024optimality,Oufkir2024Oct}. We first find the exponents for non-signaling strategies. Then we deduce the exponents in the shared randomness and entanglement assisted scenarios using the approximations algorithms of \cite{berta2024optimality}. It turns out that the approximation guarantees of \cite{berta2024optimality} are strong enough to ensure that non-signaling and shared randomness (resp.~shared entanglement) have the same exponents in the classical (resp.~classical-quantum) setting. This reduces the problem to finding  the  exponents for non-signaling strategies, which correspond to the partially smoothed max-information of the channel~\cite{fang2019quantum,cao2024channel}.

In the following, we present some techniques we used to find the exponents of channel simulation with non-signaling strategies. We prove the converse for strong converse exponent using a type of Hölder's inequality in terms of R\'enyi divergences \cite{LWD2016strong}. Moreover, the converse for error exponent is already proven by \cite{LiYao2021reliable}. The achievability for the error exponent involves constructing the smoothing channel based on the actual channel, an arbitrary replacer channel, and the relation between them.   
Establishing the strong converse achievability exponent turns out to be challenging. To address this, we employ a technique that involves an auxiliary channel (see e.g.~\cite{Haroutunian2008Feb}) and construct the smoothing channel depending on this auxiliary channel. We carry out the analysis using a fidelity-relative entropy inequality of \cite{Li2024-oper}, approximations based on Chebyshev's inequality  and  certain entropic continuity bounds. 

\renewcommand{\arraystretch}{1.5}
\begin{table}[t!]
    \centering
    \begin{tabular}{>{\centering\arraybackslash}m{5.5cm}|c|c}
        \hline
          & \textbf{Error Exponent} & \textbf{Strong Converse Exponent} \\
        \hline
        \textbf{Classical Channel Simulation} (SR, EA, NS) & $\frac{1}{2} \sup_{\alpha \geq 0} \alpha \big(r-{I}_{1+\alpha}(W)   \big)$ & $\sup_{\frac{1}{2} \leq \alpha \leq 1} \frac{1-\alpha}{\alpha} \big({I}_{\alpha}(W)-r   \big)$ \\
        \hline
        \textbf{CQ Channel Simulation} \; \; \; (EA, NS) & $\frac{1}{2} \sup_{\alpha \geq 0} \alpha \big(r-\widetilde{I}_{1+\alpha}(W)   \big)$ & $\sup_{\frac{1}{2} \leq \alpha \leq 1} \frac{1-\alpha}{\alpha} \big(\widetilde{I}_{\alpha}(W)-r   \big)$ \\
        \hline
    \end{tabular}
       \caption{Overview of Results: Exponents of classical and classical-quantum (CQ) channel simulation under the purified distance and various forms of assistance\,--\,shared randomness (SR), entanglement assistance (EA) and non-signaling (NS). }
    \label{tab:results}
\end{table}

The remainder of this paper is organized as follows. In Section \ref{subsec:formu}, we formulate the problem of channel simulation in both the classical and classical-quantum settings, considering shared randomness, entanglement-assisted, and non-signaling strategies. We then present two equivalent expressions for the minimal simulation error under the purified distance. Section \ref{sec:prem} provides the necessary preliminaries and notations for this work. In Section \ref{sec:NS:SC}, we establish the strong converse exponent for classical and classical-quantum channel simulation. Section \ref{sec:EE} focuses on deriving the error exponent for these simulations. Finally, in Section \ref{sec:conclusion}, we conclude the paper with a discussion of open questions and future directions.


\subsection{Problem formulation}
\label{subsec:formu}

\paragraph{Classical channel simulation.} For a classical channel $W: \mathcal{X} \rightarrow \mathcal{Y}$, a shared randomness assisted simulation scheme for $W$ consists of an encoding channel $\mathcal{E}: \mathcal{X} \times \mathcal{S} \rightarrow
\mathcal{M}$, a decoding channel $\mathcal{D}: \mathcal{S} \times \mathcal{M} \rightarrow \mathcal{Y}$ and a shared random variable $p_S$, the corresponding simulation channel $\widetilde{W}$ can be described as 
\begin{equation}
\widetilde{W}(y|x)=\sum_{s \in \mathcal{S}} p_S(s)\sum_{i \in \mathcal{M}} \mathcal{E}(i|x,s)\mathcal{D}(y|i,s).
\end{equation}
The classical communication cost of this simulation is $M=|\cM|$ 
and the performance of this simulation is characterized by the channel purified distance 
$P(W,\widetilde{W})\coloneqq\max\limits_{x \in \mathcal{X}} 
P(W(\cdot|x),\widetilde{W}(\cdot|x))$. We denote the optimal performance among all simulation schemes with a fixed classical communication cost as $\epsilon^{\rm{SR}}(W,M)$, i.e.,
\begin{equation}
    \epsilon^{\rm{SR}}(W,M)\coloneqq\min_{\widetilde{W} \in \Delta_M} P(W,\widetilde{W}),
\end{equation}
where $\Delta_M$ is the set of all simulation schemes of communication size $M$. The reverse Shannon theorem \cite{BSST2002entanglement} implies that when the classical communication  rate $r$ is strictly larger than the classical capacity $C(W)$,
$\epsilon^{\rm{SR}}(W^{\otimes n},e^{nr})$ converges to $0$ exponentially and when the rate $r$ is strictly below the channel capacity $C(W)$, $\epsilon^{\rm{SR}}(W^{\otimes n},e^{nr})$ converges to $1$ exponentially. The corresponding exponential decay rates are called the error exponent and the strong  converse exponent  for shared randomness assisted classical channel simulation, respectively. 

A non-signaling assisted simulation channel $\widetilde{W}$ for a classical channel $W: \mathcal{X} \rightarrow \mathcal{Y}$ is given by 
\begin{equation}
\widetilde{W}(y|x)=\sum_{i \in \mathcal{M}} N(i,y|x,i),
\end{equation}
where the channel $N: \mathcal{X}\times \mathcal{M} \rightarrow \mathcal{M} \times \mathcal{Y}$ is non-signaling, i.e.,
\begin{align}
\sum_{y \in \mathcal{Y}} N(i,y|x,j)&=N(i|x),  &\forall x\in \cX,\; \forall i,j \in \cM, \\
\sum_{i \in \mathcal{M}} N(i,y|x,j)&=N(j|y), &\forall j\in \cM,\; \forall x \in \cX,\; \forall y\in \cY.
\end{align}
The classical communication cost of this simulation is $M=|\cM|$. 
The paper~\cite{cao2024channel}  proved that the minimal amount of classical communication $M^{\rm{NS}}(W,\epsilon)$ to achieve one-shot non-signaling assisted simulation for $W$ within tolerance $\epsilon$ can be characterized by the channel's partially smoothed max-information exactly, i.e., \footnote{All logarithms in this paper are taken to be natural logarithms, and information is measured in nats.}
\begin{equation}
\label{equ:relation}
\log M^{\rm{NS}}(W,\epsilon)=I_{\rm{max}}^\epsilon(W).
\end{equation}
Based on Eq.~(\ref{equ:relation}), the optimal error $\epsilon^{\rm{NS}}(W,M)$ among all non-signaling simulation schemes for $W$ with fixed classical communication size $M$ equals to the inverse function of $I_{\rm{max}}^\epsilon(W)$. So the optimal non-signaling channel simulation error $\epsilon^{\rm{NS}}(W,M)$ is the solution to the following program \cite{cubitt2011zero,cao2024channel}:
\begin{align}
    \epsilon^{\rm{NS}}(W,M) &\coloneqq \inf_{\widetilde{W},\, q} \left\{ P\big(W, \widetilde{W}\big) \,\middle|\, \widetilde{W} \text{ channel}, \widetilde{W}(y|x) \le M q_x  \; \forall x\in \cX,\; \forall y\in \cY,\sum_{x\in \cX} q_x=1 \right\}\label{ns-pur-program-cc}
    \\&= \sup_{p_X}\inf_{\widetilde{W},\, q} \left\{ P\big(W \circ p_X, \widetilde{W} \circ p_X\big) \,\middle|\, \widetilde{W} \text{ channel}, \widetilde{W}(y|x) \le M q_x  \; \forall x\in \cX,\; \forall y\in \cY,\sum_{x\in \cX} q_x=1 \right\},\label{ns-pur-program-cc-alt}
\end{align}
where the maximization is over input probability distribution $p_X$,  $W \circ p_X(\cdot) \coloneqq \sum_{x\in \cX} p_X(x) W(\cdot|x)$ is the output probability distribution  and the last equality is proven in Lemma \ref{lem:cl:eps:minimax}.

\paragraph{Quantum channel simulation.} For a quantum channel $\mathcal{N}_{A \rightarrow B}$ from Alice to Bob, a CPTP map $\mathcal{M}_{A \rightarrow B}$ is an entanglement-assisted simulation for $\mathcal{N}_{A \rightarrow B}$ if it consists of using a shared entangled state, applying encoding channel at Alice side, sending classical information of size $M$ to Bob, and at last applying decoding channel at Bob's side. The performance of this scheme is measured by the channel purified distance
$P(\mathcal{M}_{A \rightarrow B},\mathcal{N}_{A \rightarrow B}) \coloneqq \sup_{\phi_{RA}} P(\mathcal{M}_{A \rightarrow B}(\phi_{RA}), \mathcal{N}_{A \rightarrow B}(\phi_{RA}))$ where the optimization is over pure states and $R\simeq A$. We denote the optimal performance among all simulation schemes with a classical communication cost $M$ as $\epsilon^{\rm{EA}}(\mathcal{N}_{A \rightarrow B},M)$, i.e.,
\begin{equation}
    \epsilon^{\rm{EA}}(\mathcal{N}_{A \rightarrow B},M)\coloneqq\min_{\mathcal{M}_{A \rightarrow B} \in \Xi_M} P(\mathcal{M}_{A \rightarrow B},\mathcal{N}_{A \rightarrow B}),
\end{equation}
where $\Xi_M$ is the set of all entanglement assisted simulation schemes with communication size $M$.
For a fixed classical communication cost rate $r \geq 0$, the error exponent and the strong converse exponent for entanglement-assisted quantum channel simulation are the exact rate of exponential decay under which $\epsilon^{\rm{EA}}(\mathcal{N}^{\otimes n}_{A \rightarrow B},e^{nr})$ converges  to  $0$ and $1$, respectively.

A non-signaling assisted simulation channel $\widetilde{N}_{A \rightarrow B}$ for $\mathcal{N}_{A \rightarrow B}$ can be constructed as
\begin{equation}
\widetilde{N}_{A \rightarrow B}=\Pi_{AY \rightarrow BX} \circ \id_{X \rightarrow Y},
\end{equation}
where $\id_{X \rightarrow Y}$ is the identity channel and $\Pi_{AY \rightarrow BX}$ is a  non-signaling quantum super-channel, i.e.,
\begin{align}
    \text{Tr}_X J_{\Pi}&=\frac{I_A}{|A|}\otimes \text{Tr}_{AX} J_\Pi, \\
\text{Tr}_B J_{\Pi}&=\frac{I_{Y}}{|Y|}\otimes \text{Tr}_{BY} J_\Pi,
\end{align}
where $J_{\Pi}$ is the Choi matrix of $\Pi$. The  communication cost of this simulation is $M=|X|=|Y|$. 
The paper~\cite{fang2019quantum} showed that the minimal amount of classical communication $M^{\rm{NS}}(\mathcal{N}_{A \rightarrow B},\epsilon)$ to achieve one-shot non-signaling assisted simulation for $\mathcal{N}_{A \rightarrow B}$ within tolerance $\epsilon$ can be characterized by the channel's partially smoothed max-information exactly, i.e., 
\begin{equation}
\label{equ:relation1}
\log M^{\rm{NS}}(\mathcal{N}_{A \rightarrow B},\epsilon)=I_{\rm{max}}^\epsilon(\mathcal{N}_{A \rightarrow B}).
\end{equation}
Hence, the large-deviation behaviour of the optimal error $\epsilon^{\rm{NS}}(\mathcal{N}_{A \rightarrow B},M)$ in non-signaling assisted simulation for $\mathcal{N}_{A \rightarrow B}$ can be described by the inverse function of 
$I_{\rm{max}}^\epsilon(\mathcal{N}_{A \rightarrow B})$ completely. More precisely, for a classical-quantum channel $W = \{W_x\}_{x\in \cX}$, the optimal non-signaling channel simulation error  $\epsilon^{\rm{NS}}(W,M)$ is the solution to the following semi-definite program \cite{fang2019quantum}:
\begin{align}
    \epsilon^{\rm{NS}}(W,M) &\coloneqq \inf_{\widetilde{W},\, \sigma}\left\{ P\big(W, \widetilde{W}\big) \middle| \widetilde{W} \text{ classical-quantum channel}, \widetilde{W}_x \mle M \sigma  \; \forall x\in \cX, \tr{\sigma}=1\right\} \label{ns-pur-program-cq}
    \\ &=\sup_{p_X} \inf_{\widetilde{W},\, \sigma}\left\{ P\big(W\circ p_X, \widetilde{W}\circ p_X\big) \middle| \widetilde{W} \text{ classical-quantum channel}, \widetilde{W}_x \mle M \sigma  \; \forall x\in \cX, \tr{\sigma}=1\right\},\label{ns-pur-program-cq-alt}
\end{align}
where the maximization is over input probability distribution $p_X$,  $W\circ p_X \coloneqq \sum_{x\in \cX} p_X(x) \proj{x}\otimes W_x$  and the last equality is proven in Lemma \ref{lem:cq:eps:minimax}. 


\subsection{Preliminaries}\label{sec:prem}

\subsubsection{Notation}

For a finite alphabet set $\mathcal{X}$, we denote the set of probability measures on $\mathcal{X}$ and the size of $\mathcal{X}$ as $\mathcal{P}(\mathcal{X})$ and $|\mathcal{X}|$, respectively. The notation $I_{\mathcal{X}}$ is used for the identity measure on $\mathcal{X}$. The support of a probability measure $p$ is denoted by $\supp(p)$. A classical channel 
$W: \mathcal{X} \rightarrow \mathcal{Y}$ is a stochastic matrix from $\mathcal{X}$ to $\mathcal{Y}$. We let $W(\cdot|x)$
represent the output probability measure of $W$ corresponding to input signal $x$. Let $W:\mathcal{X} \rightarrow \mathcal{Y}$ be a classical channel, we denote the set of the classical channels whose supports of the output measures are contained in those of 
the output measures of $W$ as $S_W$, i.e,
\begin{equation}
    S_W\coloneqq\left\{V~|~\supp(V(\cdot|x)) \subset \supp(W(\cdot|x)),~\forall x\in \mathcal{X} \right\}.
\end{equation}
For a classical channel $W: \mathcal{X} \rightarrow \mathcal{Y}$ and $p \in \mathcal{P}(\mathcal{X})$, we denote 
$W\circ p$ as the jointly probability distribution $W\circ p(x,y)\coloneqq p(x)W(y|x)$. 

In quantum information theory, every quantum system is associated with a Hilbert space. Let $\mathcal{H}$ be a finite dimensional Hilbert space.  We 
use $\mathcal{L}(\mathcal{H})$ and $\mathcal{P}_0(\mathcal{H})$ for the set of linear operators and the set of positive semidefinite operators on $\mathcal{H}$. 
Quantum states are  positive semidefinite 
operators with trace $1$. The set of quantum states on $\mathcal{H}$  is denoted by $\mathcal{S}(\mathcal{H})$. The set of positive quantum states on $\mathcal{H}$  is denoted by $\cS_+(\cH)$. 
When $\mathcal{H}$ is associated with a system $A$, $\mathcal{L}(\mathcal{H})$ and $\mathcal{S}(\mathcal{H})$
are written as  $\mathcal{L}(A)$ and $\mathcal{S}(A)$, respectively. For two self-adjoint operators $\rho$ and $\sigma$, $\sigma \mle \rho$ stands for $\rho-\sigma$ positive semidefinite. 
The purified distance is a commonly used distance to measure the closeness of two quantum states. It is defined as 
\begin{equation}
    P(\rho, \sigma)\coloneqq\sqrt{1-F(\rho, \sigma)}, \quad \rho, \sigma \in \mathcal{S}(\mathcal{H}),
\end{equation}
where $F(\rho, \sigma)\coloneqq\|\sqrt{\rho}\sqrt{\sigma}\|^2_1$ is the fidelity function.

A quantum channel $\mathcal{N}_{A \rightarrow B}$ is a completely positive and trace-preserving linear map from 
$\mathcal{L}(A)$ to $\mathcal{L}(B)$. A classical-quantum channel $W: \mathcal{L}(\mathcal{X}) \rightarrow \mathcal{L}(\mathcal{H})$ is a quantum channel which maps the input signals in $\mathcal{X}$ into the quantum states on $\mathcal{H}$. It can be expressed as 
\begin{equation}
    W(\cdot)=\sum_{x \in \mathcal{X}}\bra{x}(\cdot)\ket{x} W_x,
\end{equation}
or equivalently $W = \{W_x\}_{x\in \cX}$. 
The channel purified distance between two classical-quantum  channels $W = \{W_x\}_{x\in \cX}$ and $\widetilde{W}= \{\widetilde{W}_x\}_{x\in \cX}$ is 
\begin{equation}
    P(W,\widetilde{W})\coloneqq\max\limits_{x \in \mathcal{X}} 
P(W_x,\widetilde{W}_x).
\end{equation}

Let $H$ be a self-adjoint operator with spectral projections $\Pi_1,\ldots,\Pi_{v(H)}$ and $v(H)$ is the number of different eigenvalues of $H$. Then the pinching channel associated with $H$ is defined as
\begin{equation}
    \mathcal{E}_H(\cdot)\coloneqq\sum^{v(H)}_{i=1} \Pi_i(\cdot)\Pi_i.
\end{equation}
The pinching inequality~\cite{Hayashi2002optimal} tells that for any positive semidefinite operator $\sigma$, we have
\begin{equation}
\sigma \mle v(H)\mathcal{E}_H(\sigma).
\end{equation}

\subsubsection{Information divergences}

The classical relative entropy and R\'enyi relative entropy are basic tools in information theory. For $p, q \in \mathcal{P}(\mathcal{X})$, the  classical relative entropy is defined as 
\begin{equation}
D(p\|q)\coloneqq \begin{cases}
\sum_{x \in \mathcal{X}}p(x) \log \frac{p(x)}{q(x)} & \text{ if }\supp(p)\subseteq\supp(q), \\
+\infty                        & \text{ otherwise.}
                  \end{cases}
\end{equation}
Moreover, for $p, q \in \mathcal{P}(\mathcal{X})$ and $\alpha \in(0,\infty) \setminus\{1\}$, the R\'enyi relative entropy  is defined as 
\begin{equation}
D_\alpha(p\|q)\coloneqq \begin{cases}
\frac{1}{\alpha-1} \log\sum_{x \in \mathcal{X}} p(x)^\alpha q(x)^{1-\alpha} & \text{ if }\supp(p)\subseteq\supp(q)~\text{or}\, \alpha \in (0,1), \\
+\infty                        & \text{ otherwise.}
                  \end{cases}
\end{equation}
There are many different kinds of quantum generalizations of  classical relative entropy and R\'enyi relative entropy~\cite{belavkin1982c,MosonyiOgawa2017strong,MDSFT2013on,Petz1986quasi,Umegaki1954conditional,WWY2014strong}.  The most commonly used among these are the Umegaki relative entropy~\cite{Umegaki1954conditional} and the sandwiched R\'enyi relative entropy~\cite{MDSFT2013on,WWY2014strong}.
\begin{definition}
\label{definition:sand}
Let $\alpha\in(0,+\infty)\setminus\{1\}$, $\rho\in\mc{S}(\mc{H})$ and $\sigma\in\mc{P}_0(\mc{H})$.
When $\alpha >1$ and $\supp(\rho)\subseteq\supp(\sigma)$ or $\alpha\in (0,1)$ and $\supp(\rho)\not\perp\supp(\sigma)$, the sandwiched R\'enyi divergence of order $\alpha $
is defined as
\begin{equation}
    \widetilde{D}_{\alpha}(\rho \| \sigma)\coloneqq\frac{1}{\alpha-1} \log \widetilde{Q}_{\alpha}(\rho \| \sigma),
\quad\text{with}\ \
\widetilde{Q}_{\alpha}(\rho \| \sigma)=\operatorname{Tr}\big({\sigma}^{\frac{1-\alpha}{2\alpha}} \rho {\sigma}^{\frac{1-\alpha}{2\alpha}}\big)^\alpha;
\end{equation}
otherwise, we set $\widetilde{D}_{\alpha}(\rho \| \sigma)=+\infty$. 
\end{definition}
When $\alpha \rightarrow 1$,  $\widetilde{D}_\alpha(\rho\|\sigma)$ converges to the Umegaki relative entropy $D(\rho\|\sigma)$ defined as follows
\begin{equation}
D(\rho\|\sigma)\coloneqq \begin{cases}
\tr{\rho \log\rho}-\tr{\rho\log \sigma} & \text{ if }\supp(\rho)\subseteq\supp(\sigma), \\
+\infty                        & \text{ otherwise.}
                  \end{cases}
\end{equation}
When $\alpha \rightarrow \infty$, $\widetilde{D}_\alpha(\rho\|\sigma)$ converges to the max-relative entropy $D_{\rm{max}}(\rho \|\sigma)$ defined as follows
\begin{equation}
D_{\rm{max}}(\rho\|\sigma)\coloneqq\inf\{ \lambda~|~\rho \mle \exp (\lambda) \sigma\}.
\end{equation}
Let $\rho_{AB}$ be a bipartite state, the mutual information and the sandwiched R\'enyi mutual information are defined, respectively as
\begin{align}
I(A:B)_\rho&\coloneqq\inf_{\sigma_B\in \cS(B)} D(\rho_{AB} \| \rho_A \ox \sigma_B), \\
\widetilde{I}_{\alpha}(A:B)_\rho&\coloneqq\inf_{\sigma_B\in \cS(B)} \widetilde{D}_\alpha(\rho_{AB} \| \rho_A \ox \sigma_B).
\end{align}
For quantum channel $\mathcal{N}_{A \rightarrow B}$, the channel's sandwiched R\'enyi mutual information is 
\begin{equation}
\widetilde{I}_{\alpha}(\mathcal{N}_{A \rightarrow B})\coloneqq\sup_{\phi_{RA}}\inf_{\sigma_B\in \cS(B)}\widetilde{D}_{\alpha}(\mathcal{N}_{A \rightarrow B}(\phi_{RA})\|\phi_R \ox \sigma_B)
\end{equation}
where the maximization is over pure states $\phi_{RA}$ and $R\simeq A$. 
\\For classical-quantum channel $W:\cX\rightarrow \cS(B)$, the channel's sandwiched R\'enyi mutual information is 
\begin{equation}
\widetilde{I}_{\alpha}(W)\coloneqq\sup_{p_X\in \cP(\cX)}\inf_{\sigma\in \cS(B)} \widetilde{D}_\alpha(W\circ p_X\| p_X \otimes \sigma)=\inf_{\sigma\in \cS(B)} \sup_{x\in \cX} \widetilde{D}(W_x\|\sigma)
\end{equation}
where the last equality holds for $\alpha \in [\frac{1}{2},+\infty)$ \cite[Proposition 4.2]{MosonyiOgawa2017strong}.
\\Similarly, for classical channel $W:\mathcal{X} \rightarrow \mathcal{Y}$, the channel's R\'enyi mutual information is \cite{Verdu}
\begin{equation}
I_{\alpha}(W)\coloneqq\sup_{p_X\in \cP(\cX)}\inf_{q_Y\in \cP(\cY)} D_\alpha(W\circ p_X\| p_X \times q_Y)=\inf_{q_Y\in \cP(\cY)} \sup_{x\in \cX} D(W(\cdot|x)\|q_Y).
\end{equation}
In the following proposition, we collect some properties of the R\'enyi information quantities (see e.g., \cite{Beigi2013sandwiched,Frank2013Dec,WWY2014strong,MDSFT2013on,Hayashi2016Oct,MosonyiOgawa2017strong}).
\begin{proposition}
\label{prop:mainpro}
Let $\rho \in \mc{S}(\mc{H})$ and $\sigma \in \mc{P}_0(\mc{H})$. The sandwiched R\'enyi
divergence satisfies the following properties.
\begin{enumerate}[(i)]
  \item Monotonicity in R\'enyi parameter: if $0\leq \alpha \leq \beta$, then
      $\widetilde{D}_{\alpha}(\rho \| \sigma) \leq  \widetilde{D}_{\beta}(\rho \| \sigma)$;
  \item Monotonicity in $\sigma$: if $\sigma' \mge \sigma$, then $\widetilde{D}_{\alpha}(\rho \| \sigma') \leq \widetilde{D}_{\alpha}(\rho \| \sigma)$,
      when $\alpha \in [\frac{1}{2},+\infty)$;
  \item Variational representation: if $\rho$ commutes with $\sigma$, the sandwiched
      R\'enyi relative entropy~(equivalently, the classical R\'enyi relative entropy) has the following variational representation
      \beq
      D_{\alpha}(\rho \| \sigma)= \begin{cases}
         \min\limits_{\tau \in \mc{S}(\mc{H})} \big\{D(\tau \| \sigma)
         -\frac{\alpha}{\alpha-1}D(\tau \| \rho)\big\}, & \alpha \in (0,1), \\
         \max\limits_{\tau \in \mc{S}(\mc{H})} \big\{D(\tau \| \sigma)
         -\frac{\alpha}{\alpha-1}D(\tau \| \rho)\big\}, & \alpha \in (1,+\infty);
      \end{cases}
      \eeq 
  \item Data processing inequality: letting $\mc{N}$ be a CPTP map from $\mc{L}(\mc{H})$ to $\mc{L}(\mc{H}')$, we have
      \beq
      \widetilde{D}_{\alpha}(\mc{N}(\rho) \| \mc{N}(\sigma)) \leq \widetilde{D}_{\alpha}(\rho \| \sigma),
      \eeq
      when $\alpha \in [\frac{1}{2},+\infty)$;
  \item Additivity of sandwiched R\'enyi mutual information: for any $\rho_{AB} \in \mathcal{S}(AB)$, $\sigma_{A'B'} \in \mathcal{S}(A'B')$  and $\alpha \in [\frac{1}{2}, +\infty)$, we have
      \begin{equation}
       \widetilde{I}_{\alpha}(AA':BB')_{\rho\ox \sigma} =\widetilde{I}_{\alpha}(A:B)_\rho+\widetilde{I}_{\alpha}(A':B')_\sigma.
      \end{equation}
\end{enumerate}
\end{proposition}


\section{Strong converse exponent}\label{sec:NS:SC}

\subsection{Converse for non-signaling assisted classical-quantum channel simulation}
\label{subsec:converse}

In this section, we derive a lower bound for the strong converse exponent, our main result is as follows.

\begin{proposition}\label{prop:SCE-cvs}
For any classical-quantum channel $W$ and $r \geq 0$, we have that for all $n\in \mathbb{N}$ and $\alpha\in [\frac{1}{2},1]$:
\begin{equation}
 1-\epsilon^{\rm{NS}}(W^{\otimes n}, e^{nr})
\le \exp\left(-n\frac{1-\alpha}{\alpha} (\widetilde{I}_{\alpha}(W)-r)  \right),
\end{equation}
which gives the asymptotic bound on the strong converse exponent:
\begin{equation}
\label{equ:converse}
    \lim_{n \rightarrow \infty} -\frac{1}{n} \log \big( 1-\epsilon^{\rm{NS}}(W^{\otimes n},e^{nr})\big)\geq \sup_{\frac{1}{2}\leq \alpha \leq 1} \frac{1-\alpha}{\alpha} (\widetilde{I}_{\alpha}(W)-r).
\end{equation}
\end{proposition}

\begin{proof}
Recall that the non-signaling channel simulation error probability under the purified distance~(\eqref{ns-pur-program-cq-alt}) is 
\begin{align}
&\epsilon^{\rm{NS}}(W^{\otimes n}, e^{nr})\\
&=\sup_{p_{X^n}} \inf_{\widetilde{W},\, \sigma^n}\left\{ P\big(W^{\ox n}\circ p_{X^n}, \widetilde{W}\circ p_{X^n}\big) \,\middle|\, \widetilde{W} \text{ CQ channel},\, \widetilde{W}_{x^n} \mle e^{nr} \sigma^n  \;\;  \forall x^n\in \cX^{n},\, \tr{\sigma^n}=1\right\}. \label{equ:errordefiniton}
\end{align}
Then, from Lemma~\ref{lem:ldw}, for any $\alpha \in (\frac{1}{2},1)$ and $p_{X^n}\in \mathcal{P}(\mathcal{X}^{n})$, $\sigma_{B^n} \in \mathcal{S}(B^n)$ and classical-quantum channel $\widetilde{W}$ satisfy the constraints in \eqref{equ:errordefiniton},
we have  
   \begin{align}
   \label{equ:fi}
\frac{\alpha}{1-\alpha}\log F(W^{\otimes n} \circ p_{X^n}, \widetilde{W} \circ p_{X^n})  
    \leq  &-\widetilde{I}_{\alpha}(X^n:B^n)_{W^{\otimes n} \circ p_{X^n}}+ \widetilde{I}_{\beta}(X^n:B^n)_{\widetilde{W} \circ p_{X^n}}, \\
    \leq &-\widetilde{I}_{\alpha}(X^n:B^n)_{W^{\otimes n} \circ p_{X^n}}+nr,
    \end{align}
where the third line is because that $\widetilde{W} \circ p_{X^n} \mle e^{nr} p_{X^n} \otimes \sigma_{B^n}$. Observe that this inequality holds also for $\alpha=\frac{1}{2}$. 

Because Eq.~(\ref{equ:fi}) holds for any  $p_{X^n}$, $\sigma_{B^n}$ and $\widetilde{W}$ satisfy the restricted condition in \eqref{equ:errordefiniton}, we can get from \eqref{equ:errordefiniton} and using the inequality $1-\sqrt{1-x}\le x$ for $x\in[0,1]$:
\begin{align}
\label{equ:in}
 1-\epsilon^{\rm{NS}}(W^{\otimes n}, e^{nr})&= \inf_{p_{X^n}\in \mathcal{P}(\mathcal{X}^n)} \sup_{\sigma_{B^n}\in \mathcal{S}(B^n)} \sup_{\widetilde{W}:\widetilde{W}\circ p_{X^n} \mle e^{nr} p_{X^n}\otimes \sigma_{B^n} } 1-\sqrt{1-F(W^{\otimes n} \circ p_{X^n}, \widetilde{W} \circ p_{X^n}) }
\\&\le  \inf_{p_{X^n}\in \mathcal{P}(\mathcal{X}^n)} \sup_{\sigma_{B^n}\in \mathcal{S}(B^n)} \sup_{\widetilde{W}:\widetilde{W}\circ p_{X^n} \mle e^{nr} p_{X^n}\otimes \sigma_{B^n} }  F(W^{\otimes n} \circ p_{X^n}, \widetilde{W} \circ p_{X^n})
\\&\leq\inf_{p_{X^n}\in \mathcal{P}(\mathcal{X}^n)} \exp\left( \frac{1-\alpha}{\alpha} (-\widetilde{I}_{\alpha}(X^n:B^n)_{W^{\otimes n} \circ p_{X^n}}+nr  )\right) \\
&=\exp\left(-n\frac{1-\alpha}{\alpha} (\widetilde{I}_{\alpha}(W)-r )\right),
\end{align}
where the last equality comes from the additivity of the channel's sandwiched R\'enyi mutual information (see Proposition \ref{prop:mainpro}\,$(v)$ ). 
Eq.~(\ref{equ:in}) implies that
\begin{equation}
\label{equ:final2}
    \lim_{n \rightarrow \infty} -\frac{1}{n} \log \big(1-\epsilon^{\rm{NS}}(W^{\otimes n}, e^{nr})  \big) \geq  \frac{1-\alpha}{\alpha} (\widetilde{I}_{\alpha}(W)-r).
\end{equation}
Noticing that the Eq.~(\ref{equ:final2}) holds for any $\alpha \in [\frac{1}{2},1]$, we complete the proof.
\end{proof}

A similar argument can be used to derive the converse part for the strong converse exponent for non-signaling assisted classical channel simulation.  

\begin{proposition}\label{prop:SCE-cvs-cl}
For any classical channel $W$ and $r \geq 0$, we have that for all $n\in \mathbb{N}$ and $\alpha\in [\frac{1}{2},1]$:
\begin{equation}
 1-\epsilon^{\rm{NS}}(W^{\otimes n}, e^{nr})
\le \exp\left(-n\frac{1-\alpha}{\alpha} (I_{\alpha}(W)-r)  \right),
\end{equation}
which gives the asymptotic bound on the strong converse exponent:
\begin{equation}
    \lim_{n \rightarrow \infty} -\frac{1}{n} \log \big( 1-\epsilon^{\rm{NS}}(W^{\otimes n},e^{nr})\big)\geq \sup_{\frac{1}{2}\leq \alpha \leq 1} \frac{1-\alpha}{\alpha} (I_{\alpha}(W)-r).
\end{equation}
\end{proposition}

In the following, we move to prove achievability bounds for the non-signaling strong converse exponent in the classical setting  (see Section~\ref{sec:NS:CC:SC:ACH}) and classical-quantum setting (see Section~\ref{sec:NS:CQ:SC:ACH}).  
We study the classical and classical-quantum settings separately because the former is easier-to-present and instructive than the latter. While both proofs rely on the auxiliary channel technique and Chebyshev approximations, we are not able to reduce the classical-quantum setting to the fully classical one. For this reason, we use more techniques in the classical-quantum setting such as the fidelity-relative entropy inequality of \cite{Li2024-oper}, the pinching technique and a continuity bound of \cite{Fannes1973Dec,Audenaert2007Jun,petz2007quantum}. In particular the prefactor we obtain in the classical setting is of order $e^{-\cO(\sqrt{n}\log(n))}$ whereas the one we obtain in the classical-quantum setting is of order $e^{-\cO(n^{4/5}\log(n))}$. We leave the question of obtaining polynomial prefactors for future work.


\subsection{Achievability for non-signaling assisted classical channel simulation}\label{sec:NS:CC:SC:ACH}

In this section, we derive the achievability part of the strong converse exponent for non-signaling assisted classical channel simulation. The achievability part of the strong converse exponent for non-signaling assisted classical-quantum channel simulation is deferred to Section \ref{sec:NS:CQ:SC:ACH}.

We first establish a general lower bound for $1-\epsilon^{\rm{NS}}(W^{\otimes n}, e^{nr})$ and the achievability part follows from this lower bound directly. The combination of the achievability part and the converse part in \ref{subsec:converse} leads to the exact strong converse exponent for non-signaling assisted classical channel simulation. Our main result is as follows.

\begin{proposition}
\label{pro:in}
For any classical channel $W:\mathcal{X} \rightarrow \mathcal{Y}$ with $W_{\rm{min}}\coloneqq\min_{x\in \cX,\, y\in \cY :\,W(y|x)>0}\{W(y|x)\}$. Let $r \ge  0$, $\eta>0$ and $n \in \mathbb{N}$, we have
\begin{align}
 1-\epsilon^{\rm{NS}}(W^{\otimes n}, e^{nr})  
 &\geq \frac{\gamma_n^2}{8} \exp \left(-4\sqrt{A(\eta,W)n}\right)(1+\eta|\mathcal{Y}|)^{-n}
 \\& \quad\cdot \inf_{\frac{1}{2}\le \alpha \le 1}\inf_{p_X\in \cP(\cX)} \sup_{q_Y\in \mathcal{P}(\mathcal{Y})} \exp \left(-\frac{n(1-\alpha)}{\alpha}(\mathbb{E}_{x\sim p_X}D_\alpha(W(\cdot|x)\|q_Y)-r)\right),
\end{align}
where $\gamma_n \coloneqq (n+1)^{-|\mathcal{X}|}$ and  $A(\eta,W)\coloneqq2\log^2|\mathcal{Y}|+\max\left\{\log^2 W_{\rm{min}}, \log^2 \frac{\eta}{1+\eta|\mathcal{Y}|}\right\}+4$.
\end{proposition}
Choosing $\eta = \frac{1}{n}$, we obtain the following non-asymptotic bound.  
\begin{corollary}\label{cor:CC:NS:NonAsymp:Ach}
    For any classical channel $W:\mathcal{X} \rightarrow \mathcal{Y}$, $r \ge  0$, and $n \in \mathbb{N}$, we have
    \begin{align}
 1-\epsilon^{\rm{NS}}(W^{\otimes n}, e^{nr}) 
 &\geq \frac{1}{8} e^{-|\cY|}e^{-|\cX|\log(n+1)}
e^{-8\log(|\cY|^2W_{\min}(n+|\cY|))\sqrt{n}}
 \\& \quad \cdot \inf_{\frac{1}{2}\le \alpha \le 1}\inf_{p_X\in \cP(\cX)} \sup_{q_Y\in \mathcal{P}(\mathcal{Y})} \exp \left(-\frac{n(1-\alpha)}{\alpha}(\mathbb{E}_{x\sim p_X}D_\alpha(W(\cdot|x)\|q_Y)-r)\right).
\end{align}
\end{corollary}

\begin{proof}[Proof of Proposition \ref{pro:in}]
From \eqref{ns-pur-program-cc-alt}, the non-signaling channel simulation error probability under the purified distance is 
\begin{align}
&\epsilon^{\rm{NS}}(W^{\otimes n}, e^{nr}) \\
&=\sup_{p_{X^n}}\inf_{\widetilde{W},\, q_{Y^n}} \left\{ P\big( W\circ p_{X^n},  \widetilde{W}\circ p_{X^n}\big) \,\middle|\, \widetilde{W} \text{ channel},  \widetilde{W}(\cdot|x^n) \le e^{nr} q_{Y^n}  \;\; \forall x^n \in \mathcal{X}^{n}, q_{Y^n} \in \mathcal{P}(\mathcal{Y}^{n}) \right\},
\end{align}
so using the inequality $1-\sqrt{1-x}\ge\frac{x}{2} $ for $x\in [0,1]$, we have that  
\begin{align}
& \sqrt{1-\epsilon^{\rm{NS}}(W^{\otimes n}, e^{nr})} \\
&\geq \frac{1}{\sqrt{2}}\inf_{p_{X^n} \in \mathcal{P}(\mathcal{X}^{n})} \sup_{q_{Y^n} \in \mathcal{P}(\mathcal{Y}^{n})} \sum_{x^n \in \mathcal{X}^{n}} p_{X^n}(x^n)\sup_{\substack{\widetilde{W}(\cdot|x^n)\in \cP(\cY^{n}):\\\widetilde{W}(\cdot|x^n) \leq e^{nr} q^{\otimes n}_{Y}}}
 \sqrt{F\big(\widetilde{W}(\cdot | x^n),W^{\otimes n}(\cdot | x^n)\big) }\\
&\geq \frac{1}{\sqrt{2}}\inf_{p_{X^n}\in \mathcal{P}(\mathcal{X}^{n})} \sup_{q_{Y} \in \mathcal{P}(\mathcal{Y})} \sum_{x^n \in \mathcal{X}^{n}} p_{X^n}(x^n) \sup_{\substack{\widetilde{W}(\cdot|x^n)\in \cP(\cY^{n}):\\\widetilde{W}(\cdot|x^n) \leq e^{nr} q^{\otimes n}_{Y}}} \sqrt{F\big(\widetilde{W}(\cdot | x^n),W^{\otimes n}(\cdot | x^n)\big)}, \label{equ:subf}
\end{align}
where we restrict $q_{Y^n}  = q^{\otimes n}_{Y}$ in the last inequality.

It is easy to see that $\sup\limits_{{\widetilde{W}(\cdot|x^n)\in \cP(\cY^{n}):\, \widetilde{W}(\cdot|x^n) \leq e^{nr} q^{\otimes n}_{Y}}}F\big(\widetilde{W}(\cdot | x^n),W^{\otimes n}(\cdot | x^n)\big)$ only depends on the type of $x^n$ (see Section \ref{sec-types} for a definition of type). For a type $t\in \cT_n(\cX)$, let $x^n(t)\in \cX^{n}$ be a sequence of type $t$.
Hence, from Eq.~(\ref{equ:subf}) we get using the notation $\gamma_n = (n+1)^{-|\mathcal{X}|}$ and $S_W\coloneqq\left\{V~|~\supp(V(\cdot|x)) \subset \supp(W(\cdot|x)),~\forall x\in \mathcal{X} \right\}$
\begin{align}
& \sqrt{1-\epsilon^{\rm{NS}}(W^{\otimes n}, e^{nr})}  \\
&\geq  \frac{1}{\sqrt{2}}\inf_{P \in \mathcal{P}(\mathcal{T}_n(\mathcal{X}))} \sup_{q_{Y} \in \mathcal{P}(\mathcal{Y})} \sum_{t \in \mathcal{T}_n(\mathcal{X})} P(t) \sup_{\substack{\widetilde{W}(\cdot|x^n(t))\in \cP(\cY^{n}):\\\widetilde{W}(\cdot|x^n(t)) \leq e^{nr} q^{\otimes n}_{Y}}}\sqrt{F\left(\widetilde{W}(\cdot | x^n(t)),W^{\otimes n}(\cdot | x^n(t))\right)} \\
&\overset{(a)}{\geq} \frac{\gamma_n}{\sqrt{2}}\inf_{t \in {\mathcal{T}}_n(\mathcal{X})} \sup_{q_Y \in \mathcal{P}(\mathcal{Y})}\sup_{\substack{\widetilde{W}(\cdot|x^n(t))\in \cP(\cY^{n}):\\\widetilde{W}(\cdot|x^n(t)) \leq e^{nr} q^{\otimes n}_{Y}}}\sqrt{F\left(\widetilde{W}(\cdot | x^n(t)),W^{\otimes n}(\cdot | x^n(t))\right)}  \\
&=\frac{\gamma_n}{\sqrt{2}}\inf_{t \in {\mathcal{T}}_n(\mathcal{X})} \sup_{q_Y \in \mathcal{P}(\mathcal{Y})}\sup_{\substack{\widetilde{W}(\cdot|x^n(t))\in \cP(\cY^{n}):\\\widetilde{W}(\cdot|x^n(t)) \leq e^{nr} q^{\otimes n}_{Y}}} \sum_{y^n \in \mathcal{Y}^{n}} \sqrt{\widetilde{W}(y^n|x^n(t))}\sqrt{W^{\otimes n}(y^n | x^n(t))} \\
&\geq \frac{\gamma_n}{\sqrt{2}}\inf_{t \in {\mathcal{T}}_n(\mathcal{X})} \sup_{q_Y\in \mathcal{P}(\mathcal{Y}) }\sup_{V \in S_{W}}\sup_{\substack{\widetilde{W}(\cdot|x^n(t))\in \cP(\cY^{n}):\\\widetilde{W}(\cdot|x^n(t)) \leq e^{nr} q^{\otimes n}_{Y}}} \sum_{y^n \in \supp(V^{\otimes n}(\cdot|x^n))} \sqrt{\widetilde{W}(y^n|x^n(t))}\sqrt{W^{\otimes n}(y^n | x^n(t))}, \label{equ:subi}
\end{align}
where $(a)$ is because that for any $P\in \mathcal{P}(\mathcal{T}_n(\mathcal{X}))$, there must exist $t \in \mathcal{T}_n(\mathcal{X})$ such that $P(t) \geq (n+1)^{-|\mathcal{X}|}=\gamma_n$ (see Section \ref{sec-types}).

Next, we use an auxiliary channel technique (see e.g., \cite{Haroutunian2008Feb}) and Chebyshev approximations (see e.g., \cite{Blahut1974Jul}). 
For any $t \in \mathcal{T}_n(\mathcal{X})$, $q_Y \in \mathcal{P}(\mathcal{Y})$ and $V \in S_W$, we use the notation $(x)_+=\max\{x,0\}$ for a real number $x$ and let 
\begin{align}
    \xi(t,V,q_Y)&\coloneqq(\mathbb{E}_{x\sim t} D(V(\cdot|x)\|q_Y)-r)_+\;,
\\\delta(t,V,q_Y)&\coloneqq2\sqrt{\frac{1}{n}\mathbb{E}_{x\sim t}\var(V(\cdot|x)\|q_Y)}\;, \quad \zeta(t,V)\coloneqq2\sqrt{\frac{1}{n}\mathbb{E}_{x\sim t}\var(V(\cdot|x)\|W(\cdot|x))}\;,
\\\mathcal{G}(t,V)&\coloneqq\left\{y^n \in \supp\left(V^{\otimes n}(\cdot|x^n(t))\right)\,\middle|~V^{\otimes n}(y^n|x^n(t)) \leq e^{n\mathbb{E}_{x\sim t} D(V(\cdot|x)\|W(\cdot |x))+n\zeta(t,V)}W^{\otimes n}(y^n|x^n(t))\right\}, \label{eq:def-G}
\\\mathcal{S}(t,V,q_Y)&\coloneqq \left\{y^n  \in \supp\left(V^{\otimes n}(\cdot|x^n(t))\right)\,\middle|~V^{\otimes n}(y^n|x^n(t)) \leq e^{nr+n\xi(t,V,q_Y)+n\delta(t,V,q_Y)} q_{Y}^{\otimes n}(y^n)\right\}, \label{eq:def-S}
\\\beta(t,q_Y, V)&\coloneqq\frac{e^{nr}-1}{e^{nr}q_{Y}^{\otimes n}\left(\mathcal{S}(t,V,q_Y)\right)-e^{-n\xi(t,V,q_Y)-n\delta(t,V,q_Y)}V^{\otimes n}\left(\mathcal{S}(t,V,q_Y)|x_n(t)\right) }\;,
\end{align}
and  construct $\widetilde{W}_{q_Y}(\cdot| x^n(t))$ as 
\begin{equation}
\label{equ:con}
\widetilde{W}_{q_Y}(y^n| x^n(t))= \begin{cases}
\beta(t,q_Y, V) \frac{V^{\otimes n}(y^n|x_n(t))}{e^{n\xi(t,V,q_Y)+n\delta(t,V,q_Y)}}+(1-\beta(t,q_Y, V))e^{nr}q_{Y}^{\otimes n}(y^n) & \text{ if}~y^n \in \mathcal{S}(t,V,q_Y), \\
 e^{nr}q_{Y}^{\otimes n}(y^n)                  & \text{ otherwise.}
\end{cases}
\end{equation}
\sloppy It is easy to verify that if  $\beta(t,q_Y, V) \in [0,1]$, we have  $\widetilde{W}_{q_Y}(\cdot| x^n(t))\in \mathcal{P}(\mathcal{Y}^{n})$
 and satisfies $\widetilde{W}_{q_Y}(\cdot| x^n(t)) \leq e^{nr} q_{Y}^{\otimes n}$. By the definition of $\cS(t,V,q_Y)$, we have that $e^{-n\xi(t,V,q_Y)-n\delta(t,V,q_Y)} V^{\otimes n}\left(\cS(t,V,q_Y)|x^n(t)\right)) \leq e^{nr} q_{Y}^{\otimes n}\left(\cS(t,V,q_Y)\right)$ so $\beta(t,q_Y, V)\ge 0$. It remains to prove $\beta(t,q_Y, V)\le 1$, we have that 
 \begin{align}
    1-\beta(t,q_Y, V) &= \frac{e^{nr}q_{Y}^{\otimes n}\left(\mathcal{S}(t,V,q_Y)\right)-e^{-n\xi(t,V,q_Y)-n\delta(t,V,q_Y)}V^{\otimes n}\left(\mathcal{S}(t,V,q_Y)|x_n(t)\right)-e^{nr}+1}{e^{nr}q_{Y}^{\otimes n}\left(\mathcal{S}(t,V,q_Y)\right)-e^{-n\xi(t,V,q_Y)-n\delta(t,V,q_Y)}V^{\otimes n}\left(\mathcal{S}(t,V,q_Y)|x_n(t)\right) }
    \\&\ge   \frac{ e^{-n\xi(t,V,q_Y)-n\delta(t,V,q_Y)}V^{\otimes n}\left(\overline{\mathcal{S}(t,V,q_Y)}|x_n(t)\right) -e^{nr}q_{Y}^{\otimes n}\left(\overline{\mathcal{S}(t,V,q_Y)}\right)}{e^{nr}q_{Y}^{\otimes n}\left(\mathcal{S}(t,V,q_Y)\right)-e^{-n\xi(t,V,q_Y)-n\delta(t,V,q_Y)}V^{\otimes n}\left(\mathcal{S}(t,V,q_Y)|x_n(t)\right) } \ge 0,
 \end{align}
 where we used that\footnote{here we use crucially that $\xi(t,V,q_Y)\ge 0$ and this is the reason why we introduce the positive part in the definition of $\xi(t,V,q_Y)$.} $1\ge e^{-n\xi(t,V,q_Y)-n\delta(t,V,q_Y)}$ in the first inequality and the definition of $\cS(t,V,q_Y)$ in the second inequality. 
 
 Moreover, since for $y^n \in \mathcal{S}(t,V,q_Y)$ we have $V^{\otimes n}(y^n|x^n(t)) \leq e^{nr+n\xi(t,V,q_Y)+n\delta(t,V,q_Y)} q_{Y}^{\otimes n}(y^n)$, we deduce that for $y^n \in \mathcal{S}(t,V,q_Y)$:
 \begin{equation}\label{eq:tildeW-V}
     \widetilde{W}_{q_Y}(y^n| x^n(t))\ge \frac{V^{\otimes n}(y^n|x_n(t))}{e^{n\xi(t,V,q_Y)+n\delta(t,V,q_Y)}},
 \end{equation}
where we use again $\beta(t,q_Y, V) \in [0,1]$.

With the above construction in hand, we can proceed to lower bound Eq.~(\ref{equ:subi}) by choosing $\widetilde{W}(\cdot| x^n(t))=\widetilde{W}_{q_Y}(\cdot| x^n(t))$ from \eqref{equ:con}
\begin{align}
& \sqrt{1-\epsilon^{\rm{NS}}(W^{\otimes n}, e^{nr})} \\
&\geq \frac{\gamma_n}{\sqrt{2}}\inf_{t \in {\mathcal{T}}_n(\mathcal{X})} \sup_{q_Y \in \mathcal{P}(\mathcal{Y})}\sup_{V \in S_{W}} \sum_{y^n \in \supp(V^{\otimes n}(\cdot|x^n))} \sqrt{\widetilde{W}_{q_Y}(y^n| x^n(t))}\sqrt{W^{\otimes n}(y^n | x^n(t))} \\
&\geq \frac{\gamma_n}{\sqrt{2}}\inf_{t \in {\mathcal{T}}_n(\mathcal{X})} \sup_{q_Y\in \mathcal{P}(\mathcal{Y}) }\sup_{V \in S_{W}} \sum_{y^n \in \mathcal{G} \mathcal{S}(t,V,q_Y) } V^{\otimes n}(y^n | x^n(t))\sqrt{\frac{W^{\otimes n}(y^n | x^n(t))}{V^{\otimes n}(y^n | x^n(t))}}\cdot\sqrt{\frac{\widetilde{W}_{q_Y}(y^n| x^n(t))}{V^{\otimes n}(y^n | x^n(t))}} \\
&\geq \frac{\gamma_n}{\sqrt{2}}\inf_{t \in {\mathcal{T}}_n(\mathcal{X})} \sup_{q_Y\in \mathcal{P}(\mathcal{Y}) }\sup_{V \in S_{W}} e^{-\frac{n}{2}\mathbb{E}_{x\sim t} D(V(\cdot|x)\|W(\cdot |x))-\frac{n}{2}\zeta(t,V)}e^{-\frac{n}{2}\xi(t,V,q_Y)-\frac{n}{2}\delta(t,V,q_Y)}V^{\otimes n}\left(\cG\cS(t,V,q_Y) \middle| x^n(t)\right), \label{equ:subv}
\end{align}
where $ \cG\cS(t,V,q_Y) \coloneqq \mathcal{G}(t,V) \cap \mathcal{S}(t,V,q_Y)$ and we used the definition of $\mathcal{G}(t,V)$ \eqref{eq:def-G} and the inequality \eqref{eq:tildeW-V}  in the last inequality. 

From the union bound, Lemma~\ref{lem:cheg} and Lemma~\ref{lem:ches} we deduce that 
\begin{align}
    V^{\otimes n}\left( \cG\cS(t,V,q_Y)\right) &= 1-  V^{\otimes n}\left( \overline{\mathcal{G}(t,V)} \cup \overline{\mathcal{S}(t,V,q_Y)}\right) 
    \\&\ge 1-  V^{\otimes n}\left( \overline{\mathcal{G}(t,V)}\right)  -V^{\otimes n}\left( \overline{\mathcal{S}(t,V,q_Y)}\right) 
    \ge \frac{1}{2}.\label{eq:good-events-GS}
\end{align}
Therefore, we can proceed to 
lower bound Eq.~(\ref{equ:subv}) as follows
\begin{align}
    & \sqrt{1-\epsilon^{\rm{NS}}(W^{\otimes n}, e^{nr})}
\\
 & \overset{(a)}{\ge}\frac{\gamma_n}{2\sqrt{2}} \inf_{t \in {\mathcal{T}}_n(\mathcal{X})} \sup_{q_Y \in \mathcal{P}(\mathcal{Y})}\sup_{V \in S_{W}} \exp \left(-\frac{n}{2}\mathbb{E}_{x\sim t}D(V(\cdot|x)\|W(\cdot|x))-\sqrt{n\mathbb{E}_{x\sim t} \var(V(\cdot|x)\|W(\cdot|x))}\right) \\
&\quad \cdot \exp \left(-\frac{n}{2}(\mathbb{E}_{x\sim t}D(V(\cdot|x) \| q_Y)-r)_+-\sqrt{n\mathbb{E}_{x\sim t} \var(V(\cdot|x) \| q_Y)}\right) \\
&\overset{(b)}{=}  \frac{\gamma_n}{2\sqrt{2}} \inf_{t \in {\mathcal{T}}_n(\mathcal{X})} \sup_{q_Y\in \mathcal{P}(\mathcal{Y}) }\sup_{V \in S_{W}} \inf_{0\le  s \le 1} \exp \left(-\frac{n}{2}\mathbb{E}_{x\sim t}D(V(\cdot|x)\|W(\cdot|x))-\frac{ns}{2}(\mathbb{E}_{x\sim t}D(V(\cdot|x) \| q_Y)-r)\right) \\
&\quad \cdot \exp \left( -\sqrt{n\mathbb{E}_{x\sim t} \var(V(\cdot|x)\|W(\cdot|x))}-\sqrt{n\mathbb{E}_{x\sim t} \var(V(\cdot|x) \| q_Y)} \right) \\
&\overset{(c)}{\ge}  \frac{\gamma_n}{2\sqrt{2}} \inf_{t \in {\mathcal{T}}_n(\mathcal{X})} \sup_{q_Y\in \mathcal{P}(\mathcal{Y}) }\sup_{V \in S_{W}} \inf_{0\le  s \le 1} \exp \left(\!-\frac{n}{2}\mathbb{E}_{x\sim t}D(V(\cdot|x)\|W(\cdot|x))-\frac{ns}{2}\left(\mathbb{E}_{x\sim t}D\left(V(\cdot|x) \middle\| \tfrac{q_Y+\eta I_{\mathcal{Y}}}{1+\eta |\mathcal{Y}|}\right)\!-\!r\right)\!\right) \\
&\quad \cdot \exp \left( -\sqrt{n\mathbb{E}_{x\sim t} \var(V(\cdot|x)\|W(\cdot|x))}-\sqrt{n\mathbb{E}_{x\sim t} \var\left(V(\cdot|x) \middle\| \tfrac{q_Y+\eta I_{\mathcal{Y}}}{1+\eta |\mathcal{Y}|}\right)}\right),\label{equ:subv1}
\end{align}
where in $(a)$ we used the bound \eqref{eq:good-events-GS}; in $(b)$ we used $(a)_+ = \sup_{0\le s\le 1} s\cdot a$ for $a\in \mathbb{R}$; in $(c)$ we restrict the optimization over $\frac{q_Y+\eta I_{\mathcal{Y}}}{1+\eta |\mathcal{Y}|}\in \cP(\cY)$ for $q_Y\in \cP(\cY)$. 

The variance terms in \eqref{equ:subv1}  can be bounded using Lemma~\ref{lem:bound-var} 
\begin{align}
  \mathbb{E}_{x\sim t} \var(V(\cdot|x)\|W(\cdot|x)) &\le 2\log^2|\mathcal{Y}|+\log^2 W_{\rm{min}}+4\le A(\eta,W),
  \\ \mathbb{E}_{x\sim t}\var\left(V(\cdot|x) \middle\| \tfrac{q_Y+\eta I_{\mathcal{Y}}}{1+\eta |\mathcal{Y}|}\right)  &\le 2\log^2|\mathcal{Y}|+ \log^2 \tfrac{\eta}{1+\eta|\mathcal{Y}|}+4\le A(\eta,W).
\end{align}
Hence, from Eq.~(\ref{equ:subv1}), we obtain
\begin{align}
\sqrt{1-\epsilon^{\rm{NS}}(W^{\otimes n}, e^{nr})} 
&\geq \frac{\gamma_n}{2\sqrt{2}}\exp \left(-2\sqrt{A(\eta,W)n}\right) \inf_{t \in {\mathcal{T}}_n(\mathcal{X})} \sup_{q_Y\in \mathcal{P}(\mathcal{Y})} \sup_{V \in S_{W}} \inf_{0\le s \le  1} \\
&\quad \cdot \exp \left(-\frac{n}{2}\mathbb{E}_{x\sim t}D(V(\cdot|x)\|W(\cdot|x))-\frac{ns}{2}\left(\mathbb{E}_{x\sim t}D\left(V(\cdot|x) \middle\| \tfrac{q_Y+\eta I_{\mathcal{Y}}}{1+\eta |\mathcal{Y}|}\right)-r\right)\right)  \\
&\overset{(a)}{=} \frac{\gamma_n}{2\sqrt{2}}\exp\left(-2\sqrt{A(\eta,W)n}\right) \inf_{t \in {\mathcal{T}}_n(\mathcal{X})} \sup_{q_Y\in \mathcal{P}(\mathcal{Y})}  \inf_{0\le s \le  1} \sup_{V \in S_{W}} \\
&\quad \cdot \exp \left(-\frac{n}{2}\mathbb{E}_{x\sim t}D(V(\cdot|x)\|W(\cdot|x))-\frac{ns}{2}\left(\mathbb{E}_{x\sim t}D\left(V(\cdot|x) \middle\| \tfrac{q_Y+\eta I_{\mathcal{Y}}}{1+\eta |\mathcal{Y}|}\right)-r\right)\right)\\
&\overset{(b)}{=} \frac{\gamma_n}{2\sqrt{2}}\exp\left(-2\sqrt{A(\eta,W)n}\right) \inf_{t \in {\mathcal{T}}_n(\mathcal{X})} \sup_{q_Y\in \mathcal{P}(\mathcal{Y})}  \inf_{\frac{1}{2}\le \alpha \le  1} \sup_{V \in S_{W}} \\
&\quad \cdot \exp \left(\!-\frac{n}{2}\mathbb{E}_{x\sim t}D(V(\cdot|x)\|W(\cdot|x))-\frac{n(1-\alpha)}{2\alpha}\left(\mathbb{E}_{x\sim t}D\left(V(\cdot|x) \middle\| \tfrac{q_Y+\eta I_{\mathcal{Y}}}{1+\eta |\mathcal{Y}|}\right)\!-\!r\right)\!\right)
\\
&\overset{(c)}{=} \frac{\gamma_n}{2\sqrt{2}} \exp\left(-2\sqrt{A(\eta,W)n}\right)\inf_{t \in {\mathcal{T}}_n(\mathcal{X})} \sup_{q_Y\in \mathcal{P}(\mathcal{Y})}  \inf_{\frac{1}{2}\le \alpha \le 1} \\
&\quad \cdot \exp\left(-\frac{n(1-\alpha)}{2\alpha}\left(\mathbb{E}_{x\sim t}D_\alpha\left(W(\cdot|x)\middle\|\tfrac{q_Y+\eta I_{\mathcal{Y}}}{1+\eta |\mathcal{Y}|}\right)-r\right)\right),
\end{align}
where in $(a)$ we applied Sion's minimax theorem \cite{sion1958general}: we first move $\sup_{V \in S_{W}} \inf_{0\le s \le  1}$ inside $\exp(\cdot)$, the objective function $(s,V) \mapsto -\frac{n}{2}\mathbb{E}_{x\sim t}D(V(\cdot|x)\|W(\cdot|x))-\frac{ns}{2}\left(\mathbb{E}_{x\sim t}D\left(V(\cdot|x) \middle\| \tfrac{q_Y+\eta I_{\mathcal{Y}}}{1+\eta |\mathcal{Y}|}\right)-r\right)$ is linear in $s$ and concave in $V$ and the domains $S_W$ and $[0,1]$ are convex and compact; $(b)$ uses the change of variable $\alpha=\frac{1}{1+s}\in [\frac{1}{2},1]$ and $(c)$ is due to Proposition~\ref{prop:mainpro}\textit{(\romannumeral3)}.
Finally, by Proposition \ref{prop:mainpro}\,$(ii)$, we have
\begin{align}
&1-\epsilon^{\rm{NS}}(W^{\otimes n}, e^{nr}) \\
&\geq  \frac{\gamma^2_n}{8} \exp\left(-4\sqrt{A(\eta,W)n}\right)\inf_{t \in {\mathcal{T}}_n(\mathcal{X})} \sup_{q_Y\in \mathcal{P}(\mathcal{Y})}\inf_{\frac{1}{2}\le \alpha \le 1} \exp\left(-\frac{n(1-\alpha)}{\alpha}\left(\mathbb{E}_{x\sim t}D_\alpha\left(W(\cdot|x)\middle\|\tfrac{q_Y}{1+\eta |\mathcal{Y}|}\right)-r\right)\right) \\
&=\frac{\gamma^2_n}{8} e^{-4\sqrt{A(\eta,W)n}}\inf_{t \in {\mathcal{T}}_n(\mathcal{X})} \sup_{q_Y\in \mathcal{P}(\mathcal{Y})}\inf_{\frac{1}{2}\le \alpha \le 1} \exp \left(-\frac{n(1-\alpha)}{\alpha}\left(\mathbb{E}_{x\sim t}D_\alpha(W(\cdot|x)\|q_Y)\!+\!\log(1+\eta|\mathcal{Y}|)-r\right)\right)\\
&\geq \frac{\gamma^2_n}{8} e^{-4\sqrt{A(\eta,W)n}}(1+\eta|\mathcal{Y}|)^{-n}\inf_{t \in {\mathcal{T}}_n(\mathcal{X})} \sup_{q_Y\in \mathcal{P}(\mathcal{Y})}\inf_{\frac{1}{2}\le \alpha \le 1} \exp \left(-\frac{n(1-\alpha)}{\alpha}\big(\mathbb{E}_{x\sim t}D_\alpha(W(\cdot|x)\|q_Y)-r\big)\right) \\
&\geq \frac{\gamma^2_n}{8} e^{-4\sqrt{A(\eta,W)n}} (1+\eta|\mathcal{Y}|)^{-n}\inf_{p_X\in \cP(\cX)} \sup_{q_Y\in \mathcal{P}(\mathcal{Y})}\inf_{\frac{1}{2}\le \alpha \le 1} \exp \left(\!-\frac{n(1-\alpha)}{\alpha}\big(\mathbb{E}_{x\sim p_X}D_\alpha(W(\cdot|x)\|q_Y)-r\big)\!\right) \\
&\overset{(a)}{=}\frac{\gamma^2_n}{8} e^{-4\sqrt{A(\eta,W)n}}(1\!+\!\eta|\mathcal{Y}|)^{-n}\inf_{\frac{1}{2}\le \alpha \le 1}\inf_{p_X\in \cP(\cX)} \sup_{q_Y\in \mathcal{P}(\mathcal{Y})} \exp\left(\!-\frac{n(1-\alpha)}{\alpha}\big(\mathbb{E}_{x\sim p_X}D_\alpha(W(\cdot|x)\|q_Y)\!-\!r\big)\!\right),
\end{align}
where in (a), we use Sion's minimax theorem \cite{sion1958general} after making the change of variable $s= \frac{1-\alpha}{\alpha}$: the function $s \mapsto -s\big(\mathbb{E}_{x\sim p_X}D_{\frac{1}{1+s}}(W(\cdot|x)\|q_Y)-r\big)$ is convex~\cite[Corollary B.2]{MosonyiOgawa2017strong}, the function  $q_Y \mapsto -s\big(\mathbb{E}_{x\sim p_X}D_{\frac{1}{1+s}}(W(\cdot|x)\|q_Y)-r\big)$ is concave and the domains $[0,1]$ and $\cP(\cY)$ are convex and compact. 
\end{proof}

From Corollary~\ref{cor:CC:NS:NonAsymp:Ach}, we can directly derive the achievability part for the strong converse exponent for non-signaling assisted classical channel simulation.
\begin{proposition}\label{prop:CC:NS:Asymp:Ach}
For any  classical channel $W:\mathcal{X} \rightarrow \mathcal{Y}$ and $r \geq 0$, we have
\begin{equation} 
\lim_{n \rightarrow \infty} -\frac{1}{n} \log \big( 1-\epsilon^{\rm{NS}}(W^{\otimes n}, e^{nr}) \big) \leq \sup_{\frac{1}{2}\le\alpha\le1} \frac{1-\alpha}{\alpha} \big(I_{\alpha}(W)-r \big).
\end{equation}
\end{proposition}

\begin{proof}
From Corollary~\ref{cor:CC:NS:NonAsymp:Ach}, we have that for all  $n \in \mathbb{N}$
\begin{align}
-\frac{1}{n} \log\big( 1-\epsilon^{\rm{NS}}(W^{\otimes n}, e^{nr}) \big)  
\leq &\sup_{\frac{1}{2}\le\alpha\le1}\sup_{p_X \in \mathcal{P}(\mathcal{X})}\inf_{q_Y\in \mathcal{P}(\mathcal{Y})} \frac{1-\alpha}{\alpha}(\mathbb{E}_{x\sim p_X}D(W(\cdot|x)\|q_Y)-r)
\\
&+|\mathcal{X}|\frac{\log(n+1)}{n}+\frac{8\log(|\cY|^2W_{\min}(n+|\cY|))}{\sqrt{n}} +\frac{|\cY|+\log(8)}{n}.
\end{align}
Taking the limit  $n\rightarrow \infty$, we deduce
\begin{align}
\lim_{n \rightarrow \infty} -\frac{1}{n} \log \big( 1-\epsilon^{\rm{NS}}(W^{\otimes n}, e^{nr}) \big) &\leq \sup_{\frac{1}{2}\le\alpha\le1}\sup_{p_X \in \mathcal{P}(\mathcal{X})}\inf_{q_Y\in \mathcal{P}(\mathcal{Y})} \frac{1-\alpha}{\alpha}(\mathbb{E}_{x\sim p_X}D(W(\cdot|x)\|q_Y)-r)\\
&\overset{(a)}{=}\sup_{\frac{1}{2}\le\alpha\le1}\inf_{q_Y\in \mathcal{P}(\mathcal{Y})} \sup_{p_X \in \mathcal{P}(\mathcal{X})}\frac{1-\alpha}{\alpha}(\mathbb{E}_{x\sim p_X}D(W(\cdot|x)\|q_Y)-r) \\
&=\sup_{\frac{1}{2}\le\alpha\le1} \frac{1-\alpha}{\alpha}(\inf_{q_Y\in \mathcal{P}(\mathcal{Y})}\sup_{x \in \mathcal{X}}D(W(\cdot|x)\|q_Y)-r) \\
&=\sup_{\frac{1}{2}\le\alpha\le1} \frac{1-\alpha}{\alpha}(I_{\alpha}(W)-r),
\end{align}
where $(a)$ is from \cite[Proposition 1]{csiszar1995generalized}. 
\end{proof}

From Propositions \ref{prop:SCE-cvs-cl} and \ref{prop:CC:NS:Asymp:Ach} we deduce the  strong converse exponent for non-signaling assisted classical channel simulation.
\begin{theorem}\label{thm:CC:NS:Asymp}
For any  classical channel $W:\mathcal{X} \rightarrow \mathcal{Y}$ and $r \geq 0$, we have
\begin{equation} 
\lim_{n \rightarrow \infty} -\frac{1}{n} \log \big( 1-\epsilon^{\rm{NS}}(W^{\otimes n}, e^{nr}) \big) = \sup_{\frac{1}{2}\le\alpha\le1} \frac{1-\alpha}{\alpha} \big(I_{\alpha}(W)-r \big).
\end{equation}
\end{theorem}


\subsection{Achievability for non-signaling assisted classical-quantum channel simulation}\label{sec:NS:CQ:SC:ACH}

In this section, we establish the  achievability part of the strong converse exponent for non-signaling assisted classical-quantum channel simulation. The exact strong converse exponent then follows from this result and the converse part derived in Section~\ref{subsec:converse}. Our main result is stated as follows.
\begin{proposition}\label{prop:CQ:NS:Ach}
    For any  classical-quantum channel $W:\mathcal{X} \rightarrow \mathcal{S}(B)$ and $r \geq 0$, we have for all $n\ge |B|$
\begin{equation}\label{eq:CQ:NS:Ach:nonAsymp}
     1-\epsilon^{\rm{NS}}(W^{\otimes n}, e^{nr})\ge \frac{1}{2} e^{-|B| -|\cX|\log(n+1) - 32 n^{4/5} |B|^{1/5}\log(|B|(n+|B|)) }\inf_{\alpha\in [\frac{1}{2},1]}\exp\left(-n \frac{1-\alpha}{\alpha} \big(\widetilde{I}_{\alpha}(W)-r \big)\right),
\end{equation}
    which leads to the following asymptotic upper bound
\begin{equation} \label{eq:CQ:NS:Ach:Asymp}
\lim_{n \rightarrow \infty} -\frac{1}{n} \log \big( 1-\epsilon^{\rm{NS}}(W^{\otimes n}, e^{nr}) \big)  \leq \sup_{\alpha\in [\frac{1}{2},1]} \frac{1-\alpha}{\alpha} \big(\widetilde{I}_{\alpha}(W)-r \big).
\end{equation}
\end{proposition}

\begin{proof}
From \eqref{ns-pur-program-cq-alt}, the non-signaling channel simulation error probability under the purified distance is 
\begin{align}
&\epsilon^{\rm{NS}}(W^{\otimes n}, e^{nr}) \\
&=\sup_{p_{X^n}} \inf_{\widetilde{W},\, \sigma^n}\left\{ P\big( W^{\otimes n}\circ p_{X^n},  \widetilde{W}\circ p_{X^n}\big) \,\middle|\, \widetilde{W} \text{ CQ channel},\; \widetilde{W}_{x^n} \mle e^{nr} \sigma^n  \;\; \forall x^n\in \cX^{n},\; \tr{\sigma^n}=1\right\}. \label{eq:def-eps-ach-cq}
\end{align}
Then, we  fix  a constant $\eta, \delta>0$ and an integer $m \in \mathbb{N}$. It is obvious that there exists an integer $k_{\delta}$ such that when $k \geq k_\delta$, $m(k+1)(r-\delta)<mkr$. We let 
\begin{equation}
    \mathcal{S}_\eta(\cH)\coloneqq\left\{\frac{\sigma+\eta I_{\mathcal{H}}}{1+\eta|\mathcal{H}|}~\middle|~\sigma \in \mathcal{S}(\cH) \right\} = \left\{\sigma \in \mathcal{S}(\cH)~\middle|~ \|\sigma\|_{\min}\ge \frac{\eta }{1+\eta|\mathcal{H}|}\right\},
\end{equation}
where we denote the minimal positive eigenvalue of $\sigma$ as  $\|\sigma\|_{\min}$. 
Given $x^n \in \cX^n$ we denote $W_{x^n}^{\ox n}\coloneqq \bigotimes_{i=1}^n W_{x_i}$. 
 For any integer $n$ with the form $n=mk+l$, where $0\leq l <m-1$, $k \geq k_{\delta}$, $\big( 1-\epsilon^{\rm{NS}}(W^{\otimes n}, e^{nr}) \big) $ can lower bounded starting from the definition of $\epsilon^{\rm{NS}}(W^{\otimes n}, e^{nr})$ \eqref{eq:def-eps-ach-cq} and using the inequality $1-\sqrt{1-x}\ge\frac{x}{2}$ for $x\in [0,1]$ as follows
\begin{align}
&\sqrt{\big( 1-\epsilon^{\rm{NS}}(W^{\otimes n}, e^{nr}) \big) } \\
&\geq \frac{1}{\sqrt{2}}\inf_{p_{X^n}\in \cP(\cX^{n})} \sup_{\sigma \in \cS_\eta(B)}\sum_{x^n} p_{X^n}(x^n) \sup_{\widetilde{\rho}_{x^n} \in \mc{S}(B^n): \widetilde{\rho}_{x^n} \mle e^{nr} \sigma^{\ox n}} \sqrt{F(\widetilde{\rho}_{x^n},W_{x^n}^{\ox n})}\\
&\overset{(a)}{=} \frac{1}{\sqrt{2}} \inf_{P \in \mathcal{P}(\mathcal{T}_n(\mathcal{X}))} \sup_{\sigma \in \cS_\eta(B)}\sum_{t \in \mathcal{T}_n(\mathcal{X})} P(t) \sup_{\widetilde{\rho}_{x^n(t)} \in \mc{S}(B^n): \widetilde{\rho}_{x^n(t)} \mle e^{nr} \sigma^{\ox n}} \sqrt{F(\widetilde{\rho}_{x^n(t)},W_{x^n(t)}^{\ox n})} \\
&\overset{(b)}{\geq} \frac{\gamma_n}{\sqrt{2}}\inf_{x^n}\sup_{\sigma \in \cS_\eta(B)}\sup_{\widetilde{\rho}_{x^n} \in \mc{S}(B^n): \widetilde{\rho}_{x^n} \mle e^{nr} \sigma^{\ox n}} \sqrt{F(\widetilde{\rho}_{x^n},W_{x^n}^{\ox n})} \\
&\overset{(c)}{\geq}  \frac{\gamma_n}{\sqrt{2}}\inf_{x^{m(k+1)}}\sup_{\sigma \in \cS_\eta(B)} \sup_{\substack{\widetilde{\rho}_{x^{m(k+1)}} \in \mc{S}(B^{m(k+1)}):\\ \widetilde{\rho}_{x^{m(k+1)}} \mle e^{m(k+1)(r-\delta)} \sigma^{\ox m(k+1)}}} \sqrt{F\left(\widetilde{\rho}_{x^{m(k+1)}},W_{x^{m(k+1)}}^{\ox {m(k+1)}}\right)}
\\
&\overset{(d)}{\geq} \frac{\gamma_n}{\sqrt{2}}\inf_{x^{m(k+1)}}\sup_{\sigma \in \cS_\eta(B)} \sup_{\substack{\widetilde{\rho}_{x^{m(k+1)}} \in \mc{S}(B^{m(k+1)}):\\ \mc{E}_{\sigma^{\ox  m}}^{\ox (k+1)}\left(\widetilde{\rho}_{x^{m(k+1)}}\right) \mle e^{m(k+1)(r-\delta)} \sigma^{\ox m(k+1)}}} \sqrt{F\left(\mc{E}_{\sigma^{\ox  m}}^{\ox (k+1)}\left(\widetilde{\rho}_{x^{m(k+1)}}\right),W_{x^{m(k+1)}}^{\ox {m(k+1)}}\right)}\label{equ:fieva-pre}
\end{align}
where $(a)$ is because that the fidelity function is invariant under the action of permutation unitary channels so  $\sup_{\widetilde{\rho}_{x^n} \in \mc{S}(B^n): \widetilde{\rho}_{x^n} \mle e^{nr} \sigma^{\ox n}} F(\widetilde{\rho}_{x^n},W_{x^n}^{\ox n})$ only depends on the type of $x^n$; in $(b)$, we use that for any $P\in \mathcal{P}(\mathcal{T}_n(\mathcal{X}))$, there must exist $t \in \mathcal{T}_n(\mathcal{X})$ such that $P(t) \geq (n+1)^{-|\mathcal{X}|} \eqqcolon \gamma_n$ (see Section \ref{sec-types}); $(c)$ is due to the data processing inequality of the fidelity function; in $(d)$ we restrict the maximization over $\mc{E}_{\sigma^{\ox  m}}^{\ox (k+1)}\left(\widetilde{\rho}_{x^{m(k+1)}}\right)$ for $\widetilde{\rho}_{x^{m(k+1)}} \in \mc{S}(B^{m(k+1)})$.

Now from \eqref{equ:fieva-pre}, we apply  Lemma~\ref{lem:appen1}
\begin{align}
&\sqrt{\big( 1-\epsilon^{\rm{NS}}(W^{\otimes n}, e^{nr}) \big) } \\
    &\geq  \frac{\gamma_n}{\sqrt{2}}\inf_{x^{m(k+1)}}\sup_{\sigma \in \cS_\eta(B)} \sup_{\substack{\widetilde{\rho}_{x^{m(k+1)}} \in \mc{S}(B^{m(k+1)}):\\ \mc{E}_{\sigma^{\ox  m}}^{\ox (k+1)}\left(\widetilde{\rho}_{x^{m(k+1)}}\right) \mle e^{m(k+1)(r-\delta)} \sigma^{\ox m(k+1)}}} \sqrt{F\left(\mc{E}_{\sigma^{\ox  m}}^{\ox (k+1)}\left(\widetilde{\rho}_{x^{m(k+1)}}\right),W_{x^{m(k+1)}}^{\ox {m(k+1)}}\right)}
\\
&\geq\gamma_n   \inf_{x^{m(k+1)}}\sup_{\sigma \in \cS_\eta(B)}\sup_{\substack{\widetilde{\rho}_{x^{m(k+1)}} \in \mc{S}(B^{m(k+1)}):\\ \mc{E}_{\sigma^{\ox  m}}^{\ox (k+1)}\left(\widetilde{\rho}_{x^{m(k+1)}}\right) \mle e^{m(k+1)(r-\delta)} \sigma^{\ox m(k+1)}}}\frac{\sqrt{F\left(\mc{E}_{\sigma^{\ox  m}}^{\ox (k+1)}\left(\widetilde{\rho}_{x^{m(k+1)}}\right),\mc{E}_{\sigma^{\ox  m}}^{\ox (k+1)}\left(W_{x^{m(k+1)}}^{\ox {m(k+1)}}\right)\right)}}{\sqrt{2(m+1)^{(k+1)|B|}}} , \label{equ:fieva}
\end{align}
where we use that the number of the projections in the pinching channel $\mc{E}_{\sigma^{\ox  m}}^{\ox (k+1)}$ is at most $v(\sigma^{\ox  m})^{k+1}\le (m+1)^{(k+1)|B|}$   since the number $v(\sigma^{\ox  m})$ of different eigenvalues of $\sigma^{\ox  m}$ must be smaller than $(m+1)^{|B|}$ by the method of type~(see \eqref{equ:numbertype}).

Next, we let $\mathcal{S}^{\sigma}_m\coloneqq\{V \in \mathcal{S}(B^m)~|~[V,\sigma^{\otimes m}]=0  \}$. By Lemma \ref{lem:UB-rel-ent-V-tW}, for any $\{V_i \in \mathcal{S}^{\sigma}_m\}_{i=1}^{k+1}$ and $\sigma \in \mathcal{S}_\eta(B)$, we obtain the existence of $\tilde{\rho}$ such that $\tilde{\rho}\mle e^{m(k+1)(r-\delta)} \sigma^{\ox m(k+1)}  $ and 
 \begin{align}\label{eq:app-of-lemma}
        D\Big( \textstyle\bigotimes_{i=1}^{k+1} V_i \big\|\tilde{\rho}\Big)\le \Big(\textstyle\sum_{i=1}^{k+1}D\Big(V_i\| \sigma^{\otimes m}\Big)-m(k+1)(r-\delta)\Big)_+ + 30(k+1)^{3/4}\log(|B|^{m}\|\sigma^{\otimes m}\|_{\min})+\log(4).
    \end{align}
 To lower bound the above fidelity in \eqref{equ:fieva}, we use  the following inequality proved in \cite[Lemma 31]{Li2024-oper} : 
\begin{align}\label{eq:fid-KL}
    F(\rho, \sigma )\ge e^{-D(\tau\|\rho)}\cdot e^{-D(\tau\| \sigma)}\,, \quad \forall \rho, \sigma, \tau \in \cS(\cH). 
\end{align}
We set in \eqref{eq:fid-KL}, $\rho = \cE_{\sigma^{\otimes m}}^{\otimes (k+1)}\left(W^{\otimes m(k+1)}_{x^{m(k+1)}}\right)$, $\sigma = \cE_{\sigma^{\otimes m}}^{\otimes (k+1)}(\tilde{\rho})$ and $\tau =  V_1\otimes V_2\otimes \dots \otimes V_{k+1}$, where $V_i \in \mathcal{S}_m^{\sigma}, ~\forall i \in [k+1]$.  
The first relative entropy in \eqref{eq:fid-KL} can be expressed by the additivity property of the relative entropy:
\begin{align} 
  D\left(\bigotimes_{i=1}^{k+1} V_i\big\| \cE_{\sigma^{\otimes m}}^{\otimes (k+1)}\left(W^{\otimes m(k+1)}_{x^{m(k+1)}}\right)\right)
 &= \sum_{i=1}^{k+1} D\left( V_i\big\| \cE_{\sigma^{\otimes m}}^{}\left(W^{\otimes m}_{x^m(i)}\right)\right), \label{eq:D1}
\end{align}
where $(x^m(j))_{i} = x_{m(j-1)+i}$ for $i\in [m]$ and $j\in [k+1]$.  The second relative entropy in \eqref{eq:fid-KL} can be  bounded by the data processing inequality as follows, recall that $[V_i, \sigma^{\otimes m}]=0$ for $i\in [k+1]$:
\begin{align}
  D\left( \bigotimes_{i=1}^{k+1} V_i\big\| \cE_{\sigma^{\otimes m}}^{\otimes (k+1)}(\tilde{\rho})\right)=D\left( \cE_{\sigma^{\otimes m}}^{\otimes (k+1)}\left(\bigotimes_{i=1}^{k+1} V_i \right)\big\| \cE_{\sigma^{\otimes m}}^{\otimes (k+1)}\left(\tilde{\rho}\right)\right)  &\le D\left(\bigotimes_{i=1}^{k+1} V_i\big\| \tilde{\rho}\right).\label{eq:D21}
\end{align}
By Lemma \ref{lem:UB-rel-ent-V-tW}, we have an upper bound \eqref{eq:app-of-lemma} on the latter relative entropy:
\begin{align}
        D\Big( \textstyle\bigotimes_{i=1}^{k+1} V_i \big\|\tilde{\rho}\Big)\le \Big(\textstyle\sum_{i=1}^{k+1}D\Big(V_i\| \sigma^{\otimes m}\Big)-m(k+1)(r-\delta)\Big)_+ + 30(k+1)^{3/4}\log(|B|^{m}\|\sigma^{\otimes m}\|_{\min})+\log(4). \label{eq:D22}
    \end{align}
For $\sigma \in \cS_{\eta}(B)$ we have that $\|\sigma^{\otimes m}\|_{\min}\ge \left(\frac{\eta}{1+\eta|B|}\right)^{m}$. 
Hence we continue from \eqref{equ:fieva} and use \eqref{eq:fid-KL}, \eqref{eq:D1}, \eqref{eq:D21} and \eqref{eq:D22} to obtain 
\begin{align}
& 1-\epsilon^{\rm{NS}}(W^{\otimes n}, e^{nr})   \\
&\geq \frac{1}{2(m+1)^{(k+1)|B|}}\gamma_n^2 \exp\left(-30(k+1)^{3/4}m\log (|B|\cdot (\eta^{-1}+|B|)) -\log(4) \right) \\
&\quad \cdot \inf_{x^{m(k+1)}}\sup_{\sigma \in \cS_\eta(B)} \sup_{\{V_i \in \mathcal{S}_m^\sigma\}_{i=1}^{k+1}} \exp\left(\!-\!\sum_{i=1}^{k+1} D\left(\!V_i\big\| \cE_{\sigma^{\otimes m}}^{}\left(W^{\otimes m}_{x^m(i)}\right)\!\right)\!-\!\left(\sum_{i=1}^{k+1}  D\left( V_i\big\| \sigma^{\otimes m}\right)\!-\!m(k+1)(r-\delta)\!\right)_+\right) \\
&\overset{(a)}{=} \frac{1}{2(m+1)^{(k+1)|B|}}\gamma_n^2 \exp\left(-30(k+1)^{3/4}m\log (|B|\cdot (\eta^{-1}+|B|)) -\log(4) \right) \\
&\quad \cdot \inf_{x^{m(k+1)}}\sup_{\sigma \in \cS_\eta(B)} \sup_{\{V_i \in \mathcal{S}_m^\sigma\}_{i=1}^{k+1}} \inf_{0\le s\le 1}\exp\left(\!-\!\sum_{i=1}^{k+1} D\left(\!V_i\big\|\cE_{\sigma^{\otimes m}}\left(\!W^{\otimes m}_{x^m(i)}\!\right)\!\right)\!-\!s\sum_{i=1}^{k+1}  D\left(\!V_i\big\|\sigma^{\otimes m}\!\right)\!+\!sm(k+1)(r-\delta)\!\right)
\\
&\overset{(b)}{=} \frac{1}{2(m+1)^{(k+1)|B|}}\gamma_n^2 \exp\left(-30(k+1)^{3/4}m\log (|B|\cdot (\eta^{-1}+|B|)) -\log(4) \right) \\
&\quad \cdot \inf_{x^{m(k+1)}}\sup_{\sigma \in \cS_\eta(B)} \inf_{0\le s\le 1}\sup_{\{V_i \in \mathcal{S}_m^\sigma\}_{i=1}^{k+1}} \exp\left(\!-\!\sum_{i=1}^{k+1} D\left(\!V_i\big\|\cE_{\sigma^{\otimes m}}\left(\!W^{\otimes m}_{x^m(i)}\!\right)\!\right)\!-\!s\sum_{i=1}^{k+1}  D\left(\!V_i\big\|\sigma^{\otimes m}\!\right)\!+\!sm(k+1)(r-\delta)\!\right)
 \label{equ:inter-pre}
\end{align}  
where in $(a)$ we used $(f)_+ = \sup_{0\le s\le 1} s\cdot f$ for $f\in \mathbb{R}$ and in $(b)$ we used Sion's minimax theorem \cite{sion1958general}: the objective function inside the exponential function is linear in $s$ and concave in $\{V_i\}_{i\in [k+1]}$, and the sets $[0,1]$ and $\left(\mathcal{S}_m^\sigma\right)^{k+1}$ are convex and compact. 

Making the change of variable $s=\frac{1-\alpha}{\alpha}$, and using the variational expression of the R\'enyi divergence (see Proposition \ref{prop:mainpro}\,$(iii)$),  we get for all $i\in [k+1]$

\begin{align}
  &  \sup_{\{V_i \in \mathcal{S}_m^\sigma\}_{i=1}^{k+1}} \exp\left(-\sum_{i=1}^{k+1} D\left( V_i\big\| \cE_{\sigma^{\otimes m}}^{}\left(W^{\otimes m}_{x^m(i)}\right)\right)-\frac{1-\alpha}{\alpha}\sum_{i=1}^{k+1}  D\left( V_i\big\| \sigma^{\otimes m}\right) + \frac{1-\alpha}{\alpha}m(k+1)(r-\delta)  \right)
    \\&= \sup_{\{V_i \in \mathcal{S}_m^\sigma\}_{i=1}^{k+1}} \prod_{i=1}^{k+1}\exp\left(- D\left( V_i\big\| \cE_{\sigma^{\otimes m}}^{}\left(W^{\otimes m}_{x^m(i)}\right)\right)-\frac{1-\alpha}{\alpha}  D\left( V_i\big\| \sigma^{\otimes m}\right) + \frac{1-\alpha}{\alpha}m(r-\delta)  \right)
    \\&= \prod_{i=1}^{k+1}\sup_{V_i \in \mathcal{S}_m^\sigma} \exp\left(- D\left( V_i\big\| \cE_{\sigma^{\otimes m}}^{}\left(W^{\otimes m}_{x^m(i)}\right)\right)-\frac{1-\alpha}{\alpha}  D\left( V_i\big\| \sigma^{\otimes m}\right) + \frac{1-\alpha}{\alpha}m(r-\delta)  \right)
    \\&= \prod_{i=1}^{k+1}\exp \left(-\frac{1-\alpha}{\alpha}\left(D_{\alpha}\left(\cE_{\sigma^{\otimes m}}^{}\left(W^{\otimes m}_{x^m(i)}\right) \big\|\sigma^{\otimes m} \right) -m(r-\delta)\right) \right)
\end{align}
where we used in the last equality Proposition \ref{prop:mainpro}\,$(iii)$ and the fact  that $\cE_{\sigma^{\otimes m}}^{}\left(W^{\otimes m}_{x^m(i)}\right)$ commutes with $\sigma^{\otimes m}$. Hence we get from \eqref{equ:inter-pre}
\begin{align}
& 1-\epsilon^{\rm{NS}}(W^{\otimes n}, e^{nr})   \\
&\geq \frac{1}{8(m+1)^{(k+1)|B|}}\gamma_n^2 \exp\left(-30(k+1)^{3/4}m\log (|B|\cdot (\eta^{-1}+|B|)) \right) \\
&\quad \cdot \inf_{x^{m(k+1)}}\sup_{\sigma \in \cS_\eta(B)} \inf_{\frac{1}{2}\le \alpha\le 1} \prod_{i=1}^{k+1} \exp \left(-\frac{1-\alpha}{\alpha}\left(D_{\alpha}\left(\cE_{\sigma^{\otimes m}}^{}\left(W^{\otimes m}_{x^m(i)}\right) \big\|\sigma^{\otimes m} \right) -m(r-\delta)\right) \right) \\
&\overset{(a)}{\geq} \frac{1}{8(m+1)^{(k+1)|B|}}\gamma_n^2 \exp\left(-30(k+1)^{3/4}m\log (|B|\cdot (\eta^{-1}+|B|)) \right) \\
&\quad \cdot \inf_{x^{m(k+1)}}\sup_{\sigma \in \cS_\eta(B)} \inf_{\frac{1}{2}\le \alpha\le 1} \prod_{i=1}^{k+1} \exp \left(-\frac{1-\alpha}{\alpha}\left( \widetilde{D}_{\alpha}\left(W^{\otimes m}_{x^m(i)} \big\|\sigma^{\otimes m} \right) -m(r-\delta)\right) \right) \\
&\overset{(b)}{\geq} \frac{1}{8(m+1)^{(k+1)|B|}}\gamma_n^2 \exp\left(-30(k+1)^{3/4}m\log (|B|\cdot (\eta^{-1}+|B|)) \right) \\
&\quad \cdot \inf_{t \in \mathcal{P}(\mathcal{X})}\sup_{\sigma \in \cS_\eta(B)} \inf_{\frac{1}{2}\le \alpha\le 1} \exp \left(-m(k+1)\frac{1-\alpha}{\alpha}\left(\mathbb{E}_{x\sim t} \widetilde{D}_{\alpha}\left(W_x \big\|\sigma\right) -(r-\delta)\right) \right), \label{equ:inter}
\end{align}  
where in $(a)$ we use the data processing inequality of the sandwiched R\'enyi divergence (see Proposition \ref{prop:mainpro}\,$(iv)$) and in $(b)$ we use $ \widetilde{D}_{\alpha}\left(W^{\otimes m}_{x^m(i)} \big\|\sigma^{\otimes m} \right) = m\mathbb{E}_{x\sim t} \widetilde{D}_{\alpha}\left(W_x \big\|\sigma\right)$ where $t$ is the type of $x^m(i)$, the identity $\inf_{x_1,\dots, x_{k+1}}\sup_y \prod_{i=1}^{k+1}f(x_i,y)  = \inf_x\sup_y f(x,y)^{k+1}$ and further relax the infimum over types to an infimum over probability distributions.

For any $\sigma \in \mathcal{S}(B)$, $x \in \mathcal{X}$ and $\alpha\in[\frac{1}{2},1]$, we have by Proposition \ref{prop:mainpro}\,$(ii)$ 
\begin{equation}
\widetilde{D}_{\alpha}\left(W_x \middle\|\frac{\sigma+\eta I_B}{1+\eta|B|}\right)
\leq \widetilde{D}_{\alpha}\left(W_x \middle\|\frac{\sigma}{1+\eta|B|}\right)=\widetilde{D}_{\alpha}\left(W_x \big\| \sigma\right)+\log (1+\eta|B|).
\end{equation}
Hence, we can proceed to lower bound \eqref{equ:inter} as
\begin{align}
 &1-\epsilon^{\rm{NS}}(W^{\otimes n}, e^{nr})   \\
&\geq \frac{1}{8(m+1)^{(k+1)|B|}}\gamma_n^2 \exp\left(-30(k+1)^{3/4}m\log (|B|\cdot (\eta^{-1}+|B|)) \right) \left(1+\eta|B|\right)^{-m(k+1)} \\
&\quad \cdot \inf_{t \in \mathcal{P}(\mathcal{X})}\sup_{\sigma \in \mathcal{S}(B)} \inf_{\frac{1}{2}\le \alpha\le 1} \exp \left(-m(k+1)\frac{1-\alpha}{\alpha}\left(\mathbb{E}_{x\sim t} \widetilde{D}_{\alpha}\left(W_x \big\|\sigma\right) -(r-\delta)\right) \right) \\
&\overset{(a)}{=}\frac{1}{8(m+1)^{(k+1)|B|}}\gamma_n^2 \exp\left(-30(k+1)^{3/4}m\log (|B|\cdot (\eta^{-1}+|B|)) \right) \left(1+\eta|B|\right)^{-m(k+1)} \\
&\quad \cdot \inf_{t \in \mathcal{P}(\mathcal{X})}\sup_{\sigma \in \mathcal{S}(B)} \inf_{0\le s \le 1} \exp \left(-m(k+1)s\left(\mathbb{E}_{x\sim t} \widetilde{D}_{\frac{1}{1+s}}\left(W_x \big\|\sigma\right) -(r-\delta)\right) \right) \\
&\overset{(b)}{=}\frac{1}{8(m+1)^{(k+1)|B|}}\gamma_n^2 \exp\left(-30(k+1)^{3/4}m\log (|B|\cdot (\eta^{-1}+|B|)) \right) \left(1+\eta|B|\right)^{-m(k+1)} \\
&\quad \cdot \inf_{t \in \mathcal{P}(\mathcal{X})}\inf_{0\le s\le 1}\sup_{\sigma \in \mathcal{S}(B)}  \exp \left(-m(k+1)s\left(\mathbb{E}_{x\sim t} \widetilde{D}_{\frac{1}{1+s}}\left(W_x \big\|\sigma\right) -(r-\delta)\right) \right) \\
&\overset{(c)}{=}\frac{1}{8(m+1)^{(k+1)|B|}} \gamma_n^2 \exp\left(-30(k+1)^{3/4}m\log (|B|\cdot (\eta^{-1}+|B|)) \right) \left(1+\eta|B|\right)^{-m(k+1)} \\
&\quad \cdot \inf_{t \in \mathcal{P}(\mathcal{X})}\inf_{\frac{1}{2}\le \alpha\le 1}\sup_{\sigma \in \mathcal{S}(B)}  \exp \left(-m(k+1)\frac{1-\alpha}{\alpha}\left(\mathbb{E}_{x\sim t} \widetilde{D}_{\alpha}\left(W_x \big\|\sigma\right) -(r-\delta)\right) \right), \label{equ:inter2}
\end{align}
where in $(a)$ and $(c)$ we use the change of variable $s=\frac{1-\alpha}{\alpha}$ and in  $(b)$ we use  Sion's minimax theorem \cite{sion1958general}: the function $(\sigma,s) \mapsto s\left(\mathbb{E}_{x\sim t} \widetilde{D}_{\frac{1}{1+s}}\left(W_x \big\|\sigma\right) -(r-\delta)\right)$ is convex in $\sigma$ and concave in $s$~\cite[Corollary B.2]{MosonyiOgawa2017strong} and the domains $[0,1]$ and $\cS(B)$ are convex and compact.

Starting from \eqref{equ:inter2} and choosing $\delta=\frac{r}{k+1}$, $\eta= \frac{1}{n}$, 
$k= \lfloor\frac{n^{4/5}}{|B|^{4/5}}\rfloor-1$, and $m=\lfloor\frac{n}{k}\rfloor
$ 
we  deduce \eqref{eq:CQ:NS:Ach:nonAsymp}.

Moreover, from \eqref{equ:inter2}, we get
\begin{equation}
\label{equ:last1}
\lim_{n \rightarrow \infty} -\frac{1}{n} \log \big( 1-\epsilon^{\rm{NS}}(W^{\otimes n}, e^{nr}) \big)   \leq \sup_{\frac{1}{2}\leq \alpha \leq 1} \frac{1-\alpha}{\alpha}\left(\sup_{t \in \mathcal{P}(\mathcal{X})} \inf_{\sigma\in \mathcal{S}(B)} \mathbb{E}_{x\sim t}\widetilde{D}_\alpha(W_x\|\sigma)-r+\delta\right)+|B|\frac{\log(m+1)}{m}.
\end{equation}
Noticing that \eqref{equ:last1} holds for any $\delta>0$ and $m \in \mathbb{N}$ , by letting $m \rightarrow \infty$ then $\delta \rightarrow 0$, we get \eqref{eq:CQ:NS:Ach:Asymp}
\begin{equation}
\lim_{n \rightarrow \infty} -\frac{1}{n} \log \big( 1-\epsilon^{\rm{NS}}(W^{\otimes n}, e^{nr}) \big)   \leq \sup_{\frac{1}{2}\leq \alpha \leq 1} \frac{1-\alpha}{\alpha}\left(\widetilde{I}_\alpha(W)-r\right).
\end{equation}

\end{proof}

\begin{lemma} \label{lem:UB-rel-ent-V-tW}
Let $\omega\in \cS_+(\cH)$ and $\{V_i\}_{i=1}^k \in \cS(\cH)^k $ such that $[V_i, \omega]=0$ for all $i\in [k]$ and $|\cH|\ge 3$. 
Let  $s>0$ and  $n,k\in \mathbb{N}$ such that  $k\ge 56$. 
  Denote $V^{(k)} = \bigotimes_{i=1}^k V_i$. Let $C=  2\log^2(|\cH|)+ \log^2(\|\omega\|_{\min})+4$, 
\begin{align}
    \xi &= \left(\frac{1}{k}\sum_{i=1}^k D(V_i\| \omega)-s\right)_+, \quad \delta = \max\left\{\frac{C^{1/4}\log^{1/2}(|\cH|)}{k^{1/4}}, \frac{C^{1/3}\log^{1/3}\left(\|\omega\|_{\min}^{-1}\right)}{k^{1/3}}\right\},
    \\\Pi &= \left\{V^{(k)} \mle e^{ks+k \xi + k\delta } \omega^{\otimes k}\right\}\,,\quad 
    \beta = \frac{e^{ks}-1}{\tr{e^{ks}\omega^{\otimes k} -e^{-k\xi-k\delta}V^{(k)}}_+ }\in (0,1)\,,
    \\ \widetilde{W}^k &= \Pi \left(\beta e^{-k\xi-k\delta} V^{(k)}  +(1-\beta) e^{ks} \omega^{\otimes k}\right)\Pi +\Pi^c e^{ks} \omega^{\otimes k} \Pi^c.
\end{align}
We have that $\widetilde{W}^k \in \cS(\cH^{\otimes k})$,   $\widetilde{W}^k \mle e^{ks} \omega^{\otimes k}$ and 
    \begin{align}
        D\big( V^{(k)} \big\|\widetilde{W}^k\big)\le \left(\sum_{i=1}^k D(V_i\| \omega)-ks\right)_++  30k^{3/4}\log(|\cH|\|\omega\|_{\min}^{-1})+\log(4).
    \end{align}
\end{lemma}
\begin{proof}
$\widetilde{W}^k \in \cS(\cH^{\otimes k})$ follows from $V^{(k)},\, \omega\mge 0$ and $\beta\in (0,1)$. Moreover we have from the definition  $\Pi = \{V^{(k)} \mle e^{ks+k \xi + k\delta } \omega^{\otimes k}\}$ and $[\Pi, \omega^{\otimes k}]=0$:
\begin{align}
    \widetilde{W}^k &= \Pi (\beta e^{-k\xi-k\delta} V^{(k)}  +(1-\beta) e^{ks} \omega^{\otimes k})\Pi +\Pi^c e^{ks} \omega^{\otimes k} \Pi^c
    \mle \Pi e^{ks} \omega^{\otimes k}\Pi +\Pi^c e^{ks} \omega^{\otimes k} \Pi^c
    = e^{ks} \omega^{\otimes k}.
\end{align}
Using the inequality $\Pi  e^{ks} \omega^{\otimes k}\Pi\mge \Pi  e^{-k\xi-k\delta} V^{(k)} \Pi$ and $[V^{(k)},\omega^{\otimes k}]=0$, we have that:
\begin{align}
    \tr{V^{(k)}\log \widetilde{W}^k}&= \tr{V^{(k)}\log\left( \Pi (\beta e^{-k\xi-k\delta} V^{(k)}  +(1-\beta) e^{ks} \omega^{\otimes k})\Pi +\Pi^c e^{ks} \omega^{\otimes k} \Pi^c\right)}
    \\&= \tr{\Pi V^{(k)}\Pi \log\left( \Pi (\beta e^{-k\xi-k\delta} V^{(k)}  +(1-\beta) e^{ks} \omega^{\otimes k})\Pi \right)}
    \\&\quad + \tr{\Pi^cV^{(k)}\Pi^c\log\left(\Pi^c e^{ks} \omega^{\otimes k} \Pi^c\right)}
    \\&\ge  \tr{\Pi V^{(k)}\Pi \log\left( \Pi e^{-k\xi-k\delta} V^{(k)}  \Pi \right)}
     + \tr{\Pi^cV^{(k)}\Pi^c\log\left(\Pi^c e^{ks} \omega^{\otimes k} \Pi^c\right)}
     \\&=  \tr{\Pi V^{(k)}\Pi \log\left( \Pi  V^{(k)}  \Pi \right)} -(k\xi+k\delta) \cdot \tr{V^{(k)}\Pi}
    \\&\quad + \tr{\Pi^cV^{(k)}\Pi^c\log\left(\Pi^c  \omega^{\otimes k} \Pi^c\right)}+ks \cdot \tr{V^{(k)}\Pi^c}.\label{eq:LB-on-V-W}
\end{align}
By the continuity bound of Fannes \cite{Fannes1973Dec}, Audenaert \cite{Audenaert2007Jun}, and Petz \cite[Theorem 3.8]{petz2007quantum} we have that for the states $V^{(k)}$ and   $ \zeta^k = \frac{1}{\tr{V^{(k)}\Pi}}\Pi V^{(k)}\Pi$ on $\cH^{\otimes k}$:
\begin{align}\label{eq:continuity}
    \tr{V^{(k)}\log V^{(k)}} - \tr{\zeta^k\log \zeta^k} &\le \frac{1}{2}\big\|V^{(k)} - \zeta^k\big\|_1\log(|\cH|^k-1) +\log(2).
\end{align}
By Fuchs–van de Graaf inequality \cite{fuchs1999cryptographic}, we have 
    \begin{align}\label{eq:FdG}
      \frac{1}{2}\big\|V^{(k)} - \zeta^k\big\|_1\le \sqrt{1-F\left(V^{(k)}, \zeta^k\right)} = \sqrt{1-\tr{V^{(k)} \Pi} } = \sqrt{\tr{V^{(k)} \Pi^c}}.
    \end{align}
On the other hand, 
Since $\omega \in \cS_+(\cH)$ we have that $\Pi^c \omega^{\otimes k}\Pi^c\mge \left(\|\omega\|_{\min}\right)^k \Pi^c$ thus 
\begin{align}\label{eq:sigma-eta}
 \tr{\Pi^cV^{(k)}\Pi^c\log\left(\Pi^c  \omega^{\otimes k} \Pi^c\right)}\ge k\log\left(\|\omega\|_{\min}\right) \tr{V^{(k)}\Pi^c}.
\end{align}
Grouping everything together we get 
\begin{align}
     D\left(V^{(k)} \big\| \widetilde{W}^k \right)&=\tr{V^{(k)}\log V^{(k)}} - \tr{V^{(k)}\log \widetilde{W}^k}
     \\&\overset{(a)}{\le }\tr{\zeta^k\log \zeta^k} +\frac{1}{2}\left\|V^{(k)} - \zeta^k\right\|_1\log(|\cH|^k-1) +\log(2)
     \\& \quad - \tr{\Pi V^{(k)}\Pi \log\left( \Pi  V^{(k)}  \Pi \right)} +(k\xi+k\delta) \cdot \tr{V^{(k)}\Pi}
    \\&\quad - \tr{\Pi^cV^{(k)}\Pi^c\log\left(\Pi^c  \omega^{\otimes k} \Pi^c\right)}-ks \cdot \tr{V^{(k)}\Pi^c}
    \\&\overset{(b)}{\le }\frac{1}{\tr{V^{(k)}\Pi}}\tr{\Pi V^{(k)}\Pi\log \Pi V^{(k)}\Pi} -\log\tr{V^{(k)}\Pi}
    \\&\quad +k\sqrt{\tr{V^{(k)} \Pi^c}}\log(|\cH|) +\log(2)
     \\& \quad - \tr{\Pi V^{(k)}\Pi \log\left( \Pi  V^{(k)}  \Pi \right)} +(k\xi+k\delta) \cdot \tr{V^{(k)}\Pi}
    \\&\quad - \tr{\Pi^cV^{(k)}\Pi^c\log\left(\Pi^c  \omega^{\otimes k} \Pi^c\right)}-ks \cdot \tr{V^{(k)}\Pi^c}
    \\&\overset{(c)}{\le }\frac{\tr{V^{(k)}\Pi^c}}{\tr{V^{(k)}\Pi}}\tr{\Pi V^{(k)}\Pi\log \Pi V^{(k)}\Pi} -\log\tr{V^{(k)}\Pi}
    \\&\quad +k\sqrt{\tr{V^{(k)} \Pi^c}}\log(|\cH|) +\log(2)+k\xi +k\delta +k\log\left(\|\omega\|_{\min}^{-1}\right)\cdot \tr{V^{(k)}\Pi^c }
      \\&\overset{(d)}{\le }k\xi +k\delta+k\sqrt{\tr{V^{(k)} \Pi^c}}\log(|\cH|)   -\log\tr{V^{(k)}\Pi}
      +\log(2) 
  \\&\quad +k\log\left(\|\omega\|_{\min}^{-1}\right)\cdot \tr{V^{(k)}\Pi^c}, \label{eq:UB-D(V|W)}
\end{align}
\sloppy where in $(a)$ we use \eqref{eq:LB-on-V-W} and  \eqref{eq:continuity};
 in $(b)$ we use \eqref{eq:FdG} and the definition $ \zeta^k = \frac{1}{\tr{V^{(k)}\Pi}}\Pi V^{(k)}\Pi$; in $(c)$ we rearrange terms and use \eqref{eq:sigma-eta} and bound  $-ks \cdot \tr{V^{(k)}\Pi^c}\le 0$; in $(d) $ we bound $\frac{\tr{V^{(k)}\Pi^c}}{\tr{V^{(k)}\Pi}}\tr{\Pi V^{(k)}\Pi\log \Pi V^{(k)}\Pi}\le 0$ as $\Pi V^{(k)}\Pi\mle I$. To further bound \eqref{eq:UB-D(V|W)}, we need to control $\tr{V^{(k)} \Pi^c}$. To this end, we use  Chebyshev inequality \cite[Lemma 14]{Cheng2017Dec} 
 \begin{align}
    \tr{V^{(k)} \Pi^c}&= \tr{V^{(k)} \left\{ V^{(k)} \mge e^{ks+k\xi+k\delta} \omega^{\otimes k}\right\}} 
    \\&\le \frac{1}{\left[ks+k\xi+k\delta - D\left(V^{(k)}\big\| \omega^{\otimes k}\right)\right]^2}\var\left(V^{(k)}\big\| \omega^{\otimes k}\right).\label{eq:Chebyshev-quantum}
\end{align}
Observe that since $V^{(k)} = \bigotimes_{i=1}^k V_i$, we have that: 
\begin{align}
    D\left(V^{(k)}\big\| \omega^{\otimes n}\right) &= \sum_{i=1}^k D(V_i\| \omega)\;
  \text{ and }\;\var\left(V^{(k)}\big\| \omega^{\otimes n}\right)=\sum_{i=1}^k\var(V_i\| \omega),
\end{align}
so $ks+k\xi +k\delta- D\left(V^{(k)}\big\| \omega^{\otimes k}\right) = k\delta+(\sum_{i=1}^k D(V_i\| \omega)-ks)_+ -(\sum_{i=1}^k D(V_i\| \omega)-ks)\ge k\delta $. Hence from \eqref{eq:Chebyshev-quantum} we get
\begin{align}
    \tr{V^{(k)} \Pi^c}&\le  \frac{1}{\left[ks+k\xi+k\delta - D\left(V^{(k)}\big\| \omega^{\otimes k}\right)\right]^2}\var\left(V^{(k)}\big\| \omega^{\otimes k}\right)
    \le \frac{1}{k^2\delta^2} \cdot \sum_{i=1}^k\var(V_i\| \omega)
    \le  \frac{C}{k\delta^2}
\end{align}
where we used $\var(V_i\| \omega)\le 4+ 2\log^2(|\cH|)+ \log^2(\|\omega\|_{\min}) = C$ for all $i\in [k]$ (see Lemma \ref{lem:bound-var}).
Therefore, from \eqref{eq:UB-D(V|W)}
\begin{align}\label{eq:UB-D(V|W)-2}
 D\left(V^{(k)} \big\| \widetilde{W}^k\right)&\le  k\xi+ k\delta+ \frac{\sqrt{k}\sqrt{C}}{\delta}\log(|\cH|)-\log\left(1-\frac{C}{k\delta^2}\right) +\log(2)+ \frac{C}{\delta^2} \log\left(\|\omega\|_{\min}^{-1}\right).
\end{align}
Choosing $\delta = \max\left\{\frac{C^{1/4}\log^{1/2}(|\cH|)}{k^{1/4}}, \frac{C^{1/3}\log^{1/3}\left(\|\omega\|_{\min}^{-1}\right)}{k^{1/3}}\right\}$, we get from the upper bound \eqref{eq:UB-D(V|W)-2} for $k=\Omega(1)$ 
\begin{align}
     D\left(V^{(k)} \big\| \widetilde{W}^k\right)&\le  k\xi + k\delta+ \frac{\sqrt{k}\sqrt{C}}{\delta}\log(|\cH|)-\log\left(1-\frac{C}{k\delta^2}\right) +\log(2)+ \frac{C}{\delta^2} \log\left(\|\omega\|_{\min}^{-1}\right)
     \\&\overset{(a)}{\le} k\xi + 3k\max\left\{\frac{C^{1/4}\log^{1/2}(|\cH|)}{k^{1/4}}, \frac{C^{1/3}\log^{1/3}\left(\|\omega\|_{\min}^{-1}\right)}{k^{1/3}}\right\} +\log(4) \label{eq:UB-D(V|W)-3}
     \\&\overset{(b)}{\le} k\xi + 3k^{3/4}C^{1/4}\log^{1/2}(|\cH|)+3 k^{2/3}C^{1/3}\log^{1/3}\left(\|\omega\|_{\min}^{-1}\right) +\log(4)
     \\&\overset{(c)}{\le} k\xi + 6k^{3/4}\log(|\cH|)+3k^{3/4}\log^{1/2}(|\cH|)\log^{1/2}\left(\|\omega\|_{\min}^{-1}\right)+12k^{3/4}\log^{1/2}(|\cH|) 
     \\&\quad +3 k^{2/3}\log\left(\|\omega\|_{\min}^{-1}\right) +6  k^{2/3}\log^{2/3}(|\cH|)\log^{1/3}\left(\|\omega\|_{\min}^{-1}\right) +12 k^{2/3}\log^{1/3}\left(\|\omega\|_{\min}^{-1}\right) +\log(4)
      \\&\overset{(d)}{\le} \left(\sum_{i=1}^k D(V_i\| \omega)-ks\right)_++  30k^{3/4}\log(|\cH|\|\omega\|_{\min}^{-1})+\log(4).
\end{align}
where in $(a)$ we bounded $-\log\left(1-\frac{C}{k\delta^2}\right)\le \log(2)$ for our choice of $\delta$, $k\ge 56$ and $|\cH|\ge 3$; in $(b)$, we used $\max\{x,y\}\le x+y$ for positive real numbers $x$ and $y$; in $(c)$ we use $C^{\alpha} =(4+ 2\log^2(|\cH|)+ \log^2(\|\omega\|_{\min}))^{\alpha}\le 4+ 2\log^{2\alpha}(|\cH|)+\log^{2\alpha}(\|\omega\|_{\min}) $ for $\alpha\in (0,1)$; in $(d)$, we use the definition of $\xi =  \left(\frac{1}{k}\sum_{i=1}^k D(V_i\| \omega)-s\right)_+$ and bounds such as $\log^{1/2}(x)\log^{1/2}(y) \le\log(xy)$ for real numbers $x\ge 1$ and $y\ge 1$. 
\end{proof}


\subsection{Entanglement-assisted classical-quantum channel simulation}
\label{sec:strongcla}

In this section, we derive the strong converse exponent for shared-randomness-assisted classical channel simulation and, similarly, for entanglement-assisted classical-quantum channel simulation. We begin by establishing a relation between the optimal channel simulation errors with non-signaling assistance and those with shared randomness or entanglement assistance (see Lemma~\ref{lem:rounding}). This relation implies that the strong converse exponents for shared-randomness-assisted (resp.\ entanglement-assisted) classical (resp.\ classical-quantum) channel simulation are the same as those for non-signaling-assisted channel simulation. Hence, the strong converse exponent for shared-randomness assisted classical channel simulation (and likewise for entanglement-assisted classical-quantum channel simulation) follows naturally from the previous results for non-signaling assistance.

\begin{lemma}\label{lem:rounding}
Let $W$ be a classical or classical-quantum channel and $M\in \mathbb{N}$. We have that 
     \begin{align}
  1-\epsilon^{\rm{NS}}(W,M)\ge  1-\epsilon^{\cA}(W,M) \ge \frac{1}{2}\left(1-\frac{1}{e}\right)( 1-\epsilon^{\rm{NS}}(W,M))
\end{align}
where $\cA = \rm{SR}$ for classical channels $W$ and $\cA = \rm{EA}$ for classical-quantum channels $W$.
\end{lemma}

\begin{proof}
We focus in the proof on the classical-quantum setting, the classical setting being similar. 
From the rounding result of \cite{berta2024optimality} we have the existence of an  entanglement  assisted strategy $\widetilde{W}^{\rm{EA}} = \{\widetilde{W}^{\rm{EA}}_x\}_{x\in \cX}$ of size $M$ that simulates $W=\{W_x\}_{x\in \cX}$ such that for all $x\in \cX$ we have that:
	\begin{align}\label{eq:meta-EA}
		\widetilde{W}^{\rm{EA}}_x &\mge\left(1-\frac{1}{e}\right)	\widetilde{W}^{\rm{NS}}_x
	\end{align}
where $\widetilde{W}^{\rm{NS}}=\{\widetilde{W}^{\rm{NS}}_x\}_{x\in \cX}$ is a non-signaling strategy of size $M$ satisfying (see \eqref{ns-pur-program-cq}):
\begin{align}
    1-\epsilon^{\rm{NS}}(W,M) = \inf_{x\in \cX} 1-\sqrt{ 1- F\left(W_x,  \widetilde{W}^{\rm{NS}}_{x}  \right)}.
\end{align}
We have the following chain of (in)equalities 
\begin{align}
      1-\epsilon^{\rm{EA}}(W,M)  &\ge  \inf_{x\in \cX} 1-\sqrt{ 1- F\left(W_x,  \widetilde{W}^{\rm{EA}}_{x}  \right)}
      \\&\overset{(a)}{\ge} \inf_{x\in \cX} \frac{1}{2} F\left(W_x,  \widetilde{W}^{\rm{EA}}_{x}  \right)
      \\&\overset{(b)}{\ge}\frac{1}{2}\left(1-\frac{1}{e}\right) \inf_{x\in \cX}  F\left(W_x,  \widetilde{W}^{\rm{NS}}_{x}  \right)
        \\&\overset{(c)}{\ge} \frac{1}{2}\left(1-\frac{1}{e}\right) \inf_{x\in \cX}  1-\sqrt{ 1- F\left(W_x,  \widetilde{W}^{\rm{NS}}_{x}  \right)}
        \\&=  \frac{1}{2}\left(1-\frac{1}{e}\right)( 1-\epsilon^{\rm{NS}}(W,M)),
\end{align}
where in $(a)$ we use $1-\sqrt{1-z}\ge \frac{z}{2}$, in $(b)$ we use \eqref{eq:meta-EA} and in $(c)$ we use $z\ge 1-\sqrt{1-z}$.
\end{proof}

In the classical setting, from Lemma \ref{lem:rounding} we have that:
\begin{align}
  1-\epsilon^{\rm{NS}}(W^{\otimes n},e^{nr})\ge  1-\epsilon^{\rm{SR}}(W^{\otimes n},e^{nr}) \ge \frac{1}{2}\left(1-\frac{1}{e}\right)( 1-\epsilon^{\rm{NS}}(W^{\otimes n},e^{nr})).
\end{align}
So from Proposition \ref{prop:SCE-cvs-cl} and \ref{pro:in} we deduce that
\begin{corollary}\label{cor:SCE-ach}
    Let $W = \{W_x\}_{x\in \cX}$ be a classical channel and $r>0$.  
There is a constant $A$ and an integer $n_0$ such that for all $n\ge n_0$:   
    \begin{align}
    1-\epsilon^{\rm{SR}}(W^{\otimes n}, e^{nr}) &\ge  e^{-A\sqrt{n}\log(n)}\inf_{\frac{1}{2}\le \alpha \le 1}
         \exp\left(-n\tfrac{1-\alpha}{\alpha}\big(I_{\alpha}(W)-r\big)\right).
\end{align}
Moreover, the error exponent satisfies for $\cA \in \{\rm{SR}, \rm{EA}, \rm{NS}\}$
\begin{align}
  \lim_{n \rightarrow \infty} -\frac{1}{n} \log \big( 1-\epsilon^{\cA}(W^{\otimes n}, e^{nr}) \big)  = \sup_{\frac{1}{2}\le \alpha \le 1}\frac{1-\alpha}{\alpha}\left({{I}_{\alpha}(W)}-r\right).
\end{align}
\end{corollary}
Similarly, in the classical-quantum setting, from Lemma \ref{lem:rounding} we have that:
\begin{align}
  1-\epsilon^{\rm{NS}}(W^{\otimes n},e^{nr})\ge  1-\epsilon^{\rm{EA}}(W^{\otimes n},e^{nr}) \ge \frac{1}{2}\left(1-\frac{1}{e}\right)( 1-\epsilon^{\rm{NS}}(W^{\otimes n},e^{nr})).
\end{align}
So from Propositions \ref{prop:SCE-cvs} and \ref{prop:CQ:NS:Ach} we deduce that

\begin{corollary}\label{cor:SCE-ach-CQ}
    Let $W = \{W_x\}_{x\in \cX}$ be a classical-quantum channel and $r>0$.  The  strong converse exponent satisfies for $\cA \in \{ \rm{EA}, \rm{NS}\}$
\begin{align}
  \lim_{n \rightarrow \infty} -\frac{1}{n} \log \left( 1-\epsilon^{\cA}(W^{\otimes n}, e^{nr}) \right) =  \sup_{\frac{1}{2}\le \alpha \le 1}\frac{1-\alpha}{\alpha}\left({\widetilde{I}_{\alpha}(W)}-r\right).
\end{align}
\end{corollary}


\section{Error exponent}\label{sec:EE}

\subsection{Non-signaling assisted classical-quantum channel simulation}\label{sec:EE-NS}

In this section, we consider the error exponent of classical-quantum channel simulation with non-signaling assistance. Our main result is the following. 
\begin{proposition}\label{prop:EE-Ach-P}
   Let $W_{X\rightarrow B}$ be a classical-quantum channel and $d=|B|$.  We have for  all $r > 0$:
\begin{align}
        \epsilon^{\rm{NS}}(W^{\otimes n},e^{nr})&\le \inf_{\alpha\ge 0}\exp\left(-\frac{1}{2} n\alpha \Big(r-\widetilde{I}_{1+\alpha}(W)- \tfrac{\log(2)}{n}- \tfrac{d\log(n+1)}{n} \Big)\right).
\end{align}
Moreover we have 
    \begin{align}
   \lim_{n\rightarrow \infty} -\frac{1}{n}\log \epsilon^{\rm{NS}}(W^{\otimes n},e^{nr})&= \sup_{\alpha\ge 0} \frac{1}{2}\alpha\Big( r- \widetilde{I}_{1+\alpha}(W)\Big).
\end{align}
\end{proposition}
\begin{proof}
Recall the non-signaling channel simulation error probability under the purified distance \eqref{ns-pur-program-cq}:
    \begin{align}
         \epsilon^{\rm{NS}}(W,M) &\coloneqq \inf_{\widetilde{W},\, \sigma}\left\{ P\big(W, \widetilde{W}\big) \middle| \widetilde{W} \text{ classical-quantum channel}, \widetilde{W}_x \mle M \sigma  \; \forall x\in \cX, \tr{\sigma}=1\right\}.
    \end{align}
Fix $\sigma \in \cS(B)$ and let $\nu$ be the number of its distinct eigenvalues. Let $\Pi_x=\{ \cE_{\sigma}(W_x) \mge \tfrac{M-1}{\nu} \sigma\}$ and define $\widetilde{W}$ a classical-quantum channel as follows:
\begin{align}\label{eq:construction-smoothing-channel}
    \widetilde{W}_x = \Pi_x^c W_x \Pi_x^c + \tr{ W_x \Pi_x }\sigma.
\end{align}
Clearly $\widetilde{W}_x$ is a quantum state. Moreover, we have that  $W_x\mle \nu \cE_{\sigma}(W_x)$ so 
\begin{align}
    \widetilde{W}_x &\mle \Pi_x^c \nu \cE_{\sigma}(W_x) \Pi_x^c + \tr{ W_x \Pi_x }\sigma
    \mle \Pi_x^c (M-1)\sigma  \Pi_x^c + \Pi_x (M-1)\sigma \Pi_x+\sigma 
    = M\sigma.
\end{align}
On the other hand $ \widetilde{W}_x \mge  \Pi_x^c W_x \Pi_x^c $ so 
\begin{align}
    F\left(W_x, \widetilde{W}_x\right)&= \left(\tr{\sqrt{\sqrt{W_x} \widetilde{W}_x \sqrt{W_x}}}\right)^2
    \\&\ge  \left(\tr{\sqrt{\sqrt{W_x} \Pi_x^c W_x \Pi_x^c \sqrt{W_x}}}\right)^2
    \\&= \left(\tr{\sqrt{W_x} \Pi_x^c \sqrt{W_x} }\right)^2
     \\&= \tr{W_x \Pi_x^c  }^2. 
\end{align}
Therefore 
\begin{align}
 \widetilde{W}_x &\mle M\sigma
  \; \text{ and }\; 1-F\left(W_x, \widetilde{W}_x\right)\le 1-\tr{W_x \Pi_x^c  }^2\le2\,\tr{W_x \Pi_x}. 
\end{align}
\sloppy Using the definition of the projector $\Pi_x=\{ \cE_{\sigma}(W_x) \mge \tfrac{M-1}{\nu} \sigma\}$, we have that $\tr{W_x \Pi_x}= \tr{\cE_{\sigma}(W_x) \Pi_x} \ge \tfrac{M-1}{\nu}\tr{\sigma \Pi_x}$ thus by the data processing inequality of the sandwiched R\'enyi divergence applied with the measurement channel $\cN(\cdot)= \tr{\Pi_x(\cdot)}\proj{0}+ \tr{\Pi_x^c(\cdot)}\proj{1}$ we have that for all $\alpha>0$
\begin{align}
    \widetilde{D}_{1+\alpha}(W_x\|\sigma) &\ge     \widetilde{D}_{1+\alpha}(\cN(W_x)\|\cN(\sigma))
    \\&\ge \frac{1}{\alpha}\log \tr{W_x\Pi_x}^{1+\alpha}\tr{\sigma \Pi_x}^{-\alpha}
    \\&\ge  \frac{1}{\alpha}\log \tr{W_x\Pi_x}^{1+\alpha}(\tfrac{\nu}{M-1}\tr{W_x \Pi_x})^{-\alpha}
    \\&= \frac{1}{\alpha}\log \tr{W_x\Pi_x} + \log(\tfrac{M-1}{\nu}).
\end{align}
Therefore for all $\alpha>0$
\begin{align}
     1-F\left(W_x, \widetilde{W}_x\right)\le 2\,\tr{W_x\Pi_x} \le 2\exp\left(\alpha\big(\widetilde{D}_{1+\alpha}(W_x\|\sigma)-\log(\tfrac{M-1}{\nu})\big)\right).
\end{align}
Now, with $n$ uses of the channel, we  choose $\sigma=\sigma^{\otimes n}$ to be an iid state and construct the channel $\widetilde{W}^n$ as described in \eqref{eq:construction-smoothing-channel} with respect to $W^{\otimes n}$ and $\sigma^{\otimes n}$. 
Since $ \widetilde{W}_{x^n}^{n} \mle M\sigma^{\otimes n} =  e^{nr}\sigma^{\otimes n}$,  $(\widetilde{W}^n, \sigma^{\otimes n})$ is a  feasible point of the NS program \eqref{ns-pur-program-cq}. Note that $\nu(\sigma^{\otimes n}) \le (n+1)^{d}$. Therefore for all $\alpha\ge 0$, for all $\sigma$
\begin{align}
        \epsilon^{\rm{NS}}(W^{\otimes n},e^{nr})&\le  \max_{x^n\in \cX^n} \sqrt{1- F\left(W_{x^n}^{\otimes n} , \widetilde{W}^n_{x^n} \right)} 
          \\&\le  \max_{x^n\in \cX^n}\exp\left( \frac{1}{2}\alpha\left(\widetilde{D}_{1+\alpha}(W^{\otimes n}_{x^n}\|\sigma^{\otimes n})-\log(\tfrac{M-1}{\nu}) \right)\right)
        \\&\le  \max_{x^n\in \cX^n}\exp\left( \frac{1}{2}\alpha\left(\widetilde{D}_{1+\alpha}(W^{\otimes n}_{x^n}\|\sigma^{\otimes n})-\log(e^{nr}-1)+ d\log(n+1) \right)\right)
          \\&\le \sup_{T\in \cT_n(\cX)}\exp\left(\frac{1}{2} \alpha\left( \exs{x\sim T}{n\widetilde{D}_{1+\alpha}(W_{x}\|\sigma)}-nr+ \log(2)+ d\log(n+1) \right)\right)
           \\&\le \sup_{p\in \cP(\cX)}\exp\left(\frac{1}{2} \alpha\left( \exs{x\sim p}{n\widetilde{D}_{1+\alpha}(W_{x}\|\sigma)}-nr+ \log(2)+ d\log(n+1) \right)\right).\label{eq:UB-P-NS}
\end{align}
Since this inequality holds for all $\alpha\ge 0$ and for all $\sigma$ we deduce that 
\begin{align}
        \epsilon^{\rm{NS}}(W^{\otimes n},e^{nr})&\le \inf_{\alpha\ge 0}\inf_{\sigma\in \cS(B)}\sup_{p\in \cP(\cX)}\exp\left(\frac{1}{2} n\alpha \Big( \exs{x\sim p}{\widetilde{D}_{1+\alpha}(W_{x}\|\sigma)}-r+ \tfrac{\log(2)}{n}+ \tfrac{d\log(n+1)}{n} \Big)\right)
        \\&= \inf_{\alpha\ge 0}\exp\left(-\frac{1}{2} n\alpha \Big(r-\widetilde{I}_{1+\alpha}(W)- \tfrac{\log(2)}{n}- \tfrac{d\log(n+1)}{n} \Big)\right).
\end{align}
Finally, this achievability bound along with the converse bound of \cite{LiYao2021reliable} imply
\begin{align}
   \lim_{n\rightarrow \infty } -\frac{1}{n}\log \epsilon^{\rm{NS}}(W^{\otimes n},e^{nr})&= \sup_{\alpha\ge 0} \frac{1}{2}\alpha\Big( r- \widetilde{I}_{1+\alpha}(W)\Big).
\end{align}
\end{proof}

Using a similar proof, we can obtain the error exponent for non-signaling assisted classical channel simulation.

\begin{proposition}\label{EE-Ach-P-cl}
   Let $W$ be a classical channel.  We have for all $r>0$:
\begin{align}
        \epsilon^{\rm{NS}}(W^{\otimes n},e^{nr})&\le \inf_{\alpha\ge 0}\exp\left(-\frac{1}{2} n\alpha \Big(r-{I}_{1+\alpha}(W)- \tfrac{\log(2)}{n} \Big)\right).
\end{align}
Moreover we have
    \begin{align}
   \lim_{n\rightarrow \infty } -\frac{1}{n}\log \epsilon^{\rm{NS}}(W^{\otimes n},e^{nr})= \sup_{\alpha\ge 0} \frac{1}{2}\alpha\Big( r- I_{1+\alpha}(W)\Big).
\end{align}
\end{proposition}


\subsection{Entanglement-assisted classical-quantum channel simulation} 

In this section, we establish a relation between the optimal error of shared randomness and entanglement-assisted codes with a fixed size and their non-signaling  counterparts~(Lemma~\ref{lem:rounding-EA}). From this relation and the exact characterization of non-signaling error exponent we established in  Section \ref{sec:EE-NS}, we can deduce the 
error exponent for shared randomness  (resp.\  entanglement-assisted) classical (resp.\  classical-quantum) channel simulation. 

\begin{lemma}\label{lem:rounding-EA}
    Let $M, M'\in \mathbb{N}$ such that $M'\ge \log(2) M$ and $W$ be a classical or classical-quantum channel. We have that 
   \begin{align}
    \epsilon^{\cA}(W,M') \le \epsilon^{\rm{NS}}(W,M) +\sqrt{2} \exp\left(-\tfrac{M'}{2M}\right) 
\end{align}
where $\cA = \rm{SR}$ for classical channels $W$ and $\cA = \rm{EA}$ for classical-quantum channels $W$.
\end{lemma}
\begin{proof}
We focus in the proof on the classical-quantum setting, the classical setting being similar. 
From the rounding result of \cite{berta2024optimality} we have the existence of an entanglement-assisted strategy $\widetilde{W}^{\rm{EA}} = \{\widetilde{W}^{\rm{EA}}_x\}_{x\in \cX}$ of size $M'$ that simulates $W=\{W_x\}_{x\in \cX}$ such that for all $x\in \cX$ we have that:
	\begin{align}\label{eq:meta-EA-M-M'}
		\widetilde{W}^{\rm{EA}}_x &\mge\left[1-\left(1-\frac{1}{M}\right)^{M'}\right]	\widetilde{W}^{\rm{NS}}_x
	\end{align}
where $\widetilde{W}^{\rm{NS}}=\{\widetilde{W}^{\rm{NS}}_x\}_{x\in \cX}$ is a non-signaling strategy satisfying (see \eqref{ns-pur-program-cq}):
\begin{align}
    \epsilon^{\rm{NS}}(W,M) = \max_{x\in \cX}\sqrt{ 1- F\left(W_x,  \widetilde{W}^{\rm{NS}}_{x}  \right)}.
\end{align}
From Inequality \eqref{eq:meta-EA-M-M'} we have that 
\begin{align}
    F\left(W_x,  \widetilde{W}^{\rm{EA}}_{x}  \right) \ge \left[1-\left(1-\frac{1}{M}\right)^{M'}\right] F\left(W_x,  \widetilde{W}^{\rm{NS}}_{x}  \right).
\end{align}
Hence for $M'\ge \log(2) M$ we have that
\begin{align}
    \epsilon^{\rm{NS}}(W,M)^2 &= \max_{x\in \cX} \; 1- F\left(W_x,  \widetilde{W}^{\rm{NS}}_{x}  \right)
    \\&\ge \max_{x\in \cX}\; 1- \left[1-\left(1-\tfrac{1}{M}\right)^{M'}\right]^{-1}F\left(W_x,  \widetilde{W}^{\rm{EA}}_{x}  \right)
     \\&\overset{(a)}{\ge} \; 1- \left[1+2\left(1-\tfrac{1}{M}\right)^{M'}\right]\cdot (1- \epsilon^{\rm{EA}}(W,M')^2)
        \\&\ge  \epsilon^{\rm{EA}}(W,M')^2 -2 \left(1-\tfrac{1}{M}\right)^{M'}
         \\&\ge  \epsilon^{\rm{EA}}(W,M')^2 -2 \exp\left(-\tfrac{M'}{M}\right)
\end{align}
where in $(a)$ we use the inequality $(1-z)^{-1}\le 1+2z$ valid for $z\le \frac{1}{2}$ applied on $z=(1-\frac{1}{M})^{M'}\le \exp(-\frac{M'}{M})$. 
    Using the inequality $\sqrt{x+y}\le \sqrt{x}+\sqrt{y}$ for positive $x$ and $y$, we deduce 
 \begin{align}
    \epsilon^{\rm{EA}}(W,M') \le  \sqrt{\epsilon^{\rm{NS}}(W,M)^2 +2 \exp\left(-\tfrac{M'}{M}\right)} \le\epsilon^{\rm{NS}}(W,M) +\sqrt{2} \exp\left(-\tfrac{M'}{2M}\right) 
\end{align}
\end{proof}
Using Lemma~\ref{lem:rounding-EA}, we are able to find the shared randomness and entanglement-assisted error exponents of channel simulation.
\begin{proposition}\label{prop:EE-Ach-P-EA}
  Let $W$ be a classical or classical-quantum channel. 
    We have for all $r > 0$:
\begin{align}
    \lim_{n\rightarrow \infty } -\frac{1}{n}\log \epsilon^{\cA}(W^{\otimes n},e^{nr})&= \sup_{\alpha\ge 0} \frac{1}{2}\alpha\left(r- {\widetilde{I}_{1+\alpha}(W)}\right)
\end{align}
where $\cA \in \{\rm{SR}, \rm{EA}, \rm{NS}\}$ for classical channels $W$ and $\cA \in \{ \rm{EA}, \rm{NS}\}$  for classical-quantum channels $W$.
\end{proposition}
\begin{proof}
We focus in the proof on the classical-quantum setting, the classical setting being similar. 
From Proposition \ref{prop:EE-Ach-P}, we have that for all $\delta\ge 0$: 
\begin{align}
       \epsilon^{\rm{NS}}(W^{\otimes n} ,e^{n(r-\delta)})&\le\inf_{\alpha\ge 0}\exp\left(-\frac{1}{2} n\alpha \Big(r-\delta-\widetilde{I}_{1+\alpha}(W)- \tfrac{\log(2)}{n}- \tfrac{d\log(n+1)}{n} \Big)\right).
\end{align}
Hence by Lemma \ref{lem:rounding-EA} we deduce that for all $\delta>0$
\begin{align}
     \epsilon^{\rm{EA}}(W^{\otimes n} ,e^{nr}) \le \inf_{\alpha\ge 0}\exp\left(-\frac{1}{2} n\alpha \Big(r-\delta-\widetilde{I}_{1+\alpha}(W)- \tfrac{\log(2)}{n}- \tfrac{d\log(n+1)}{n} \Big)\right) + \sqrt{2}\exp\left(-\frac{1}{2}e^{n\delta}\right)
\end{align}
Therefore using \cite[Appendix C]{Oufkir2024Oct} we obtain 
\begin{align}
    \limsup_{n\rightarrow \infty} \frac{1}{n}\log \epsilon^{\rm{EA}}(W^{\otimes n} ,e^{nr}) &\le \inf_{\delta>0}\inf_{\alpha\ge 0} -\frac{1}{2} \alpha \Big(r-\delta-\widetilde{I}_{1+\alpha}(W) \Big)
    \\&\le \inf_{\alpha\ge 0} -\frac{1}{2} \alpha \Big(r-\widetilde{I}_{1+\alpha}(W) \Big).
\end{align}
The converse inequality follows from \cite{LiYao2021reliable}.

\end{proof}


\section{Conclusion}\label{sec:conclusion}

In this work, we establish the exact strong converse exponents and error exponents for non-signaling assisted, shared randomness assisted classical channel simulation, as well as non-signaling assisted, entanglement-assisted classical-quantum channel simulation. This is the first complete large deviation analysis for the  classical-quantum channel simulation problem. 
In contrast to previous work~\cite{LiYao2021reliable}, our characterization of the error exponents does not  exhibit a critical rate. Hence, we address the gap in the error exponent in the high-rate regime for classical-quantum channel simulation. We conjecture that there is no critical rate in the error exponent for general quantum channel simulation, and we leave this as an open question for future research. Another interesting open question is to determine the exponents of channel simulation under the diamond distance in the quantum setting.

\section*{Acknowledgment}

AO, YY and MB acknowledge funding by the European Research Council (ERC Grant Agreement No. 948139), MB acknowledges support from the Excellence Cluster - Matter and Light for Quantum Computing (ML4Q).

\printbibliography


\appendix

\section{Lemmas}
\begin{lemma}\label{lem:cl:eps:minimax}
Let $W$ be a classical channel. The non-signaling channel simulation error can be expressed as
\begin{align}
    \epsilon^{\rm{NS}}(W,M) &= 
    \inf_{\widetilde{W},\, q} \left\{ P\big(W, \widetilde{W}\big) \,\middle|\, \widetilde{W} \mathrm{ \;channel}, \widetilde{W}(y|x) \le M q_x  \; \forall x\in \cX,\; \forall y\in \cY,\sum_{x\in \cX} q_x=1 \right\}
    \\&= \sup_{p\in \cP(\cX)}\inf_{\widetilde{W},\, q} \left\{ P\big(W \circ p_X, \widetilde{W} \circ p_X\big) \,\middle|\, \widetilde{W} \mathrm{\; channel}, \widetilde{W}(y|x) \le M q_x  \; \forall x\in \cX,\; \forall y\in \cY,\sum_{x\in \cX} q_x=1 \right\}.
\end{align}
\end{lemma}
\begin{proof}
The first equality follows from \cite{cao2024channel}. In the following, we prove the second equality. 
Given two channels $W$ and $\widetilde{W}$ and a probability distribution $p$, the objective function 
\begin{align}
    P\big(W \circ p_X, \widetilde{W} \circ p_X\big) &= \sqrt{1-\left(\sum_{x,y} \sqrt{W \circ p_X(x,y)}\cdot \sqrt{\widetilde{W} \circ p_X(x,y)} \right)^2 }
    \\&= \sqrt{1-\left(\sum_{x,y} p(x) \sqrt{W(y|x)}\cdot \sqrt{\widetilde{W}(y|x)} \right)^2 }
    \\&= \sqrt{1-f(p, \widetilde{W}, q)^2 }
\end{align}
where $f(p, \widetilde{W}, q) = \sum_{x,y} p(x) \sqrt{W(y|x)}\cdot \sqrt{\widetilde{W}(y|x)}$ is linear in $p$ and concave in $(\widetilde{W}, q)$. Moreover the set of probability distribution is convex and compact. Finally, the set 
\begin{equation}
    \cS = \left\{ \big(\widetilde{W}, q\big) \,\middle|\, \widetilde{W} \text{ channel}, \widetilde{W}(y|x) \le M q_x  \; \forall x\in \cX,\; \forall y\in \cY,\sum_{x\in \cX} q_x=1 \right\}
\end{equation}
is convex so the conditions of Sion's minimax theorem \cite{sion1958general} are fulfilled and we can swap $p$ and $(\widetilde{W}, q)$ in the following 
    \begin{align}
        \sup_{p\in \cP(\cX)}\inf_{(\widetilde{W},\, q)\in \cS}   P\big(W \circ p_X, \widetilde{W} \circ p_X\big) &= \sqrt{1- \left(\inf_{p\in \cP(\cX)}\sup_{(\widetilde{W},\, q)\in \cS}f(p, \widetilde{W}, q)\right)^2 }
        \\&= \sqrt{1- \left(\sup_{(\widetilde{W},\, q)\in \cS}\inf_{p\in \cP(\cX)}f(p, \widetilde{W}, q)\right)^2 }
         \\&= \sqrt{1- \left(\sup_{(\widetilde{W},\, q)\in \cS}\min_{x\in \cX}\sum_{y} \sqrt{W(y|x)}\cdot \sqrt{\widetilde{W}(y|x)}\right)^2 }
         \\&= \inf_{(\widetilde{W},\, q)\in \cS}P\big(W, \widetilde{W}\big).
    \end{align}
\end{proof}
\begin{lemma}\label{lem:cq:eps:minimax}
Let $W=\{W_x\}_{x\in \cX}$ be a classical-quantum (CQ) channel. We have that 
    \begin{align}
   \epsilon^{\rm{NS}}(W,M) &= 
  \inf_{\widetilde{W},\, \sigma}\left\{ P\big(W, \widetilde{W}\big) \middle| \widetilde{W} \mathrm{ \; CQ\;  channel}, \widetilde{W}_x \mle M \sigma  \quad \forall x\in \cX, \tr{\sigma}=1\right\}, 
    \\ &=\sup_{p\in \cP(\cX)} \inf_{\widetilde{W},\, \sigma}\left\{ P\big(W\circ p_X, \widetilde{W}\circ p_X\big) \middle| \widetilde{W} \mathrm{\;CQ\;channel}, \widetilde{W}_x \mle M \sigma  \; \forall x\in \cX, \tr{\sigma}=1\right\},
\end{align}
\end{lemma}
\begin{proof}The first equality follows from \cite{fang2019quantum}. In the following, we prove the second equality. 
Given two classical-quantum channels $W=\{W_x\}_{x\in \cX}$ and $\widetilde{W}=\{\widetilde{W}_x\}_{x\in \cX}$ and a probability distribution $p$, the objective function 
\begin{align}
    P\big(W\circ p_X, \widetilde{W}\circ p_X\big) &= \sqrt{1-\left(\sum_{x\in \cX} p(x) \sqrt{F\left(W_x,\widetilde{W}_x\right)} \right)^2 }
    \\&= \sqrt{1-f(p, \widetilde{W}, \sigma)^2 }
\end{align}
where $f(p, \widetilde{W}, \sigma) = \sum_{x\in \cX} p(x) \sqrt{F\left(W_x,\widetilde{W}_x\right)}$ is linear in $p$ and concave in $(\widetilde{W}, \sigma)$ \cite{wilde2013quantum}. Moreover the set of probability distribution is convex and compact. Finally, the set 
\begin{equation}
    \cS = \left\{ \big(\widetilde{W}, \sigma\big) \,\middle|\, \widetilde{W} \text{ classical-quantum channel}, \widetilde{W}_x \mle M \sigma  \; \forall x\in \cX, \tr{\sigma} =1 \right\}
\end{equation}
is convex so the conditions of Sion's minimax theorem \cite{sion1958general} are fulfilled and we can swap $p$ and $(\widetilde{W}, \sigma)$ in the following 
    \begin{align}
        \sup_{p\in \cP(\cX)}\inf_{(\widetilde{W},\, \sigma)\in \cS}   P\big(W\circ p_X, \widetilde{W}\circ p_X\big) &= \sqrt{1- \left(\inf_{p\in \cP(\cX)}\sup_{(\widetilde{W},\, \sigma)\in \cS}f(p, \widetilde{W}, \sigma)\right)^2 }
        \\&= \sqrt{1- \left(\sup_{(\widetilde{W},\, \sigma)\in \cS}\inf_{p\in \cP(\cX)}f(p, \widetilde{W}, \sigma)\right)^2 }
         \\&= \sqrt{1- \left(\sup_{(\widetilde{W},\, \sigma)\in \cS}\min_{x\in \cX}\sqrt{F\left(W_x,\widetilde{W}_x\right)} \right)^2 }
         \\&= \inf_{(\widetilde{W},\sigma)\in \cS}P\big(W, \widetilde{W}\big).
    \end{align}
\end{proof}

The following lemma comes from~\cite[Proposition 2.8]{LWD2016strong}.
\begin{lemma}[\cite{LWD2016strong}]
\label{lem:ldw}
Let $\rho_{AB}, \sigma_{AB}  \in \mathcal{S}(AB)$, assuming that $\rho_{A}=\sigma_A$, then for any $\alpha \in (\frac{1}{2},1)$, we have
\begin{equation}
\frac{\alpha}{1-\alpha} \log F(\rho_{AB}, \sigma_{AB}) \leq -\widetilde{I}_{\alpha}(A:B)_\rho+\widetilde{I}_{\beta}(A:B)_\sigma,
\end{equation}
where $\frac{1}{\alpha}+\frac{1}{\beta}=2$.
\end{lemma}
\begin{lemma}[\cite{Li2024-oper}]
\label{lem:appen1}
Let $\rho, \sigma\in\mc{S}(\mc{H})$, and let $\mc{H}=\bigoplus_{i\in \mc{I}}\mc{H}_i$ decompose into a set of mutually orthogonal subspaces $\{\mc{H}_i\}_{i\in \mc{I}}$. Suppose that $\sigma=\sum\limits_{i \in \mc{I}} \sigma_i$ with $\supp(\sigma_i)\subseteq \mc{H}_i$. Then
\begin{equation}
F\left(\sum_{i \in \mc{I}} \Pi_i \rho \Pi_i, \sigma\right) \leq \sqrt{|\mc{I}|} F(\rho, \sigma),
\end{equation}
where $\Pi_i$ is the projection onto $\mc{H}_i$.
\end{lemma}

\section{Method of types}
\label{sec-types}

Let $n$ be an integer and  $\cX$ be a finite alphabet of size $|\cX|$. Let $x^n = x_1\cdots x_n$ be an element of $\cX^n$. For $x\in\cX$ we define $n(x|x^n)$ to be the number of occurrences of $x$ in the sequence $x_1\cdots x_n$: 
\begin{equation}
   n(x|x^n) = \sum_{t=1}^n \bid\{x_t=x\}. 
\end{equation}
A type $T$ is a probability distribution on $\cX$ of the form 
\begin{equation}
    T=\Big\{\frac{n_x}{n}\Big\}_{x\in \cX} \text{ where } n_x\in \mathbb{N} \text{ and } \sum_{x\in \cX} n_x=n. \label{eq:type}
\end{equation}
The set of types of alphabet $\cX$ of length $n$ is denoted $\cT_n(\cX)$. It is a finite set of size satisfying 
\begin{equation}
\label{equ:numbertype}
   |\cT_n(\cX)|\le (n+1)^{|\cX|}. 
\end{equation}
This simple bound can be proved using the simple  observation that each $n_x$ in Eq.~\eqref{eq:type} satisfies $n_x \in \{0,1, \dots, n\}$ and thus it has at most $n+1$ possibilities.
\\We say that $x^n = x_1\cdots x_n$ has type $T$ and write $x^n \sim T$ if for all $x\in \cX$ we have $ {n(x| x^n)} = {n}T_x$.
\\A probability distribution $p$ on $\cX^n$ is permutation invariant if for all permutation $\sigma \in \fS_n$, for all $ x_1\cdots x_n \in \cX^n$ we have 
\begin{equation}
    p(x_1\cdots x_n) = p(x_{\sigma_1}\cdots x_{\sigma_n}).
\end{equation}
Let  $T\in \cT_n(\cX)$ be a type and  $p$ be a permutation invariant  probability distribution on $\cX^n$.
Clearly if two sequences $x_{1}^n$ and $y_{1}^n$ have the same type $T$ then $y_{1}^n$ can be obtained by permuting the elements of $x_{1}^n$ so
\begin{equation}
    p(x^n)= p(y_1^n) , \quad \forall x^n\sim T, \; \forall y_1^n\sim T.
\end{equation}
We denote this value by $\frac{\alpha_T}{|T|}$ where $|T|$ is the number of elements in $\cX^n$ of type $T$. So we can write 
\begin{equation}
    p = \sum_{T\in \cT_n(\cX)} \alpha_T \cU_T
\end{equation}
where $\cU_T$ is the uniform probability distribution supported on $T$:
\begin{equation}
    \cU_T(x^n) = \frac{1}{|T|}\bid\{x^n\sim T\}.
\end{equation}
Since $p$ and $\{\cU_T\}_{T\in \cT_n(\cX)}$ are all probability distributions on $\cX^n$, $(\alpha_T)_{T\in \cT_n(\cX)}$ is a probability distribution on $\cT_n(\cX)$. In particular, since we have $|\cT_n(\cX)|\le (n+1)^{|\cX|}$ there is a type $T^\star$ such that 
\begin{equation}
    \alpha_{T^\star} \ge \frac{1}{ (n+1)^{|\cX|}}
\end{equation}
 because otherwise $\sum_{T\in \cT_n(\cX)}\alpha_T< \sum_{T\in \cT_n(\cX)} \frac{1}{ (n+1)^{|\cX|}} = |\cT_n(\cX)|\cdot\frac{1}{ (n+1)^{|\cX|}} \le 1$ which contradicts the fact that $(\alpha_T)_{T\in \cT_n(\cX)}$ is a probability distribution on $\cT_n(\cX)$.
 
 \section{Chebyshev approximations}
\begin{lemma}
\label{lem:cheg}
Let  $W:\mathcal{X} \rightarrow \mathcal{Y}$ be a classical channel and  $V \in \mathcal{S}_W$. For any $x^n \in \mathcal{X}^{n}$ of type $t(x^n)$, we define 
\begin{align}
    \mathcal{G}(x^n,V)\coloneqq \left\{y^n: \log \frac{V^{\otimes n}(y^n|x^n)}{W^{\otimes n}(y^n|x^n)}
\leq n\mathbb{E}_{x\sim  t(x^n)} D(V(\cdot|x)\|W(\cdot|x)) +2\sqrt{n \mathbb{E}_{x\sim  t(x^n)} \var\left(V(\cdot|x)\middle\|W(\cdot|x)\right)}\right\}.
\end{align}
Then, we have
\begin{equation}
\sum_{y^n \in \mathcal{G}(x^n,V)}V^{\otimes n}(y^n|x^n) \geq \frac{3}{4}.
\end{equation}
\end{lemma}

\begin{proof}
We regard $\log\frac{V^{\otimes n}(\cdot|x^n)}{W^{\otimes n}(\cdot|x^n)}$ as a random variable on $\mathcal{Y}^{n}$ and the probability measure on $\mathcal{Y}^{n}$ is $V^{\otimes n}(\cdot|x^n)$. By direct calculation
\begin{align}
    \mathbb{E}_{V^{\otimes n}(\cdot|x^n)}\left[\log\frac{V^{\otimes n}(\cdot|x^n)}{W^{\otimes n}(\cdot|x^n)}\right]&=n\mathbb{E}_{x\sim t(x^n)} D(V(\cdot|x)\|W(\cdot|x)), \\
\var_{V^{\otimes n}(\cdot|x^n)}\left[\log\frac{V^{\otimes n}(\cdot|x^n)}{W^{\otimes n}(\cdot|x^n)}\right]&=n \mathbb{E}_{x\sim t(x^n)} \var(V(\cdot|x)\|W(\cdot|x)).
\end{align}
Hence, by Chebyshev inequality, we have
\begin{align}
    &\sum_{y^n \in \overline{\mathcal{G}(x^n,V)}}V^{\otimes n}(y^n|x^n)
    \\&=V^{\otimes n}(\cdot|x^n)\left(\log \frac{V^{\otimes n}(\cdot|x^n)}{W^{\otimes n}(\cdot|x^n)}
>\mathbb{E}_{V^{\otimes n}(\cdot|x^n)}\left[\log\frac{V^{\otimes n}(\cdot|x^n)}{W^{\otimes n}(\cdot|x^n)}\right] +2\sqrt{\var_{V^{\otimes n}(\cdot|x^n)}\left[\log\frac{V^{\otimes n}(\cdot|x^n)}{W^{\otimes n}(\cdot|x^n)}\right]}\right)      
    \\
&\le \frac{\var_{V^{\otimes n}(\cdot|x^n)}\left[\log\frac{V^{\otimes n}(\cdot|x^n)}{W^{\otimes n}(\cdot|x^n)}\right]}{4\var_{V^{\otimes n}(\cdot|x^n)}\left[\log\frac{V^{\otimes n}(\cdot|x^n)}{W^{\otimes n}(\cdot|x^n)}\right]}=\frac{1}{4}.
\end{align}
\end{proof}
Using a similar argument in the proof of Lemma~\ref{lem:cheg}, we can obtain the following lemma. 
\begin{lemma}
\label{lem:ches}
Let  $W:\mathcal{X} \rightarrow \mathcal{Y}$ be a classical channel,  $V \in \mathcal{S}_W$ and $q_Y\in \cP(\cY)$. For any $x^n \in \mathcal{X}^{n}$ of type $t(x^n)$, we define 
\begin{align}
    \mathcal{S}(x^n,V, q_Y)\coloneqq \left\{y^n: \log \frac{V^{\otimes n}(y^n|x^n)}{q_{Y}^{\otimes n}(y^n)}
\leq n\mathbb{E}_{x\sim  t(x^n)} D(V(\cdot|x)\|q_Y) +2\sqrt{n \mathbb{E}_{x\sim  t(x^n)} \var\left(V(\cdot|x)\middle\|q_Y\right)}\right\}.
\end{align}
Then, we have
\begin{equation}
\sum_{y^n \in \mathcal{S}(x^n,V,q_Y)}V^{\otimes n}(y^n|x^n) \geq \frac{3}{4}.
\end{equation}
\end{lemma}

\section{A bound on the variance}

\begin{lemma}\label{lem:bound-var}
    Let $\rho \in \cS(\cH)$ and $\omega \in \cS(\cH)$ such that $[\rho, \omega]=0$ and $\supp(\rho) \subseteq \supp(\omega)$. We have that 
    \begin{align}
    \var(\rho\| \omega) \le 2\log^2(|\cH|)+\log^2(\|\omega\|_{\min})+4,
\end{align}
where $\|\omega\|_{\rm{min}}$ is the minimal non-zero eigenvalue of $\omega$.
\end{lemma}
\begin{proof}
   Since $[\rho, \omega]=0$, we have that 
   \begin{align}
        \var(\rho\| \omega) &\le  \tr{\rho \log^2\big(\frac{\rho}{\omega}\big) \bid\{\rho\mle \omega\}}+\tr{\rho \log^2\big(\frac{\rho}{\omega}\big) \bid\{\rho\mge \omega\}}
       \\& \le  \tr{\rho \log^2\big(\frac{1}{\rho}\big) \bid\{\rho\mle \omega\}}+\tr{\rho \log^2\big(\frac{1}{\omega}\big) \bid\{\rho\mge \omega\}}
       \\&\le 2+2\log^2(|\cH|)+2+\log^2(\|\omega\|_{\min})
   \end{align}
   where we used $\tr{\rho \log^2\rho}\le 2\log^2(|\cH|)+4$ proven below and $\rho \log^2\big(\frac{1}{\omega}\big) \mle \log^2\big(\frac{1}{\|\omega\|_{\min}}\big) \rho$.

Let $\{\lambda_i\}_i$ be the eigenvalues of $\rho$ and   $S=\{ i \in |\cH| : \lambda_i \le \frac{1}{e}\}$. The choice of the threshold $\frac{1}{e}$ is justified by the fact that the function $x\mapsto x\log^2(x)$ is concave between $[0, \frac{1}{e}]$. Moreover we can bound the size of $S^c$ as follows 
\begin{align}
    |S^c|\cdot \frac{1}{e}\le \sum_{i\in S^c} \lambda_i \le \sum_{i=1}^{|\cH|} \lambda_i=1 \Rightarrow |S^c|\le e \Rightarrow |S^c|\le 2. 
\end{align}
On the other hand using the concavity of $x\mapsto x\log^2(x)$  between $[0, \frac{1}{e}]$ we have that:
\begin{align}
    \frac{1}{|S|}\textstyle\sum_{i\in S} \lambda_i \log^2(\lambda_i) &\le \left(\frac{1}{|S|}\textstyle\sum_{i\in S}\lambda_i\right)\log^2\left(\frac{1}{|S|}\textstyle\sum_{i\in S}\lambda_i\right)
    \\&\le 2\left(\frac{1}{|S|}\textstyle\sum_{i\in S}\lambda_i\right)\log^2\left(\textstyle\sum_{i\in S}\lambda_i\right)+2\left(\frac{1}{|S|}\textstyle\sum_{i\in S}\lambda_i\right)\log^2\left(|S|\right)
    \\&\le \frac{2}{|S|}+\frac{2\log^2(|\cH|)}{|S|}
\end{align}
where the second inequality uses $\log^2(ab)=(\log a +\log b)^2\le 2\log^2(a)+2\log^2(b)$ and the third inequality uses $\sum_{i\in S}\lambda_i \le 1$, $x\log^2(x)\le 1$ for $x\in [0,1]$ and $|S|\le |\cH|$.  Therefore using again $x\log^2(x)\le 1$ for $x\in[0,1]$
\begin{align}
   \tr{\rho \log^2\rho}= \sum_{i=1}^{|\cH|} \lambda_i\log^2(\lambda_i) &= \sum_{i\in S} \lambda_i \log^2(\lambda_i)+ \sum_{i\in S^c} \lambda_i \log^2(\lambda_i)
    \\&\le |S| \cdot\left(\frac{2}{|S|}+\frac{2\log^2(|\cH|)}{|S|}\right)+ |S^c|
    \\&\le 2+2\log^2(|\cH|)+2. 
\end{align}
\end{proof}
In the general case of non-commuting $\rho$ and $\omega$, one could start from \cite[Lemma III.4]{Dupuis2019Jul} to prove such bounds.
\end{document}